%% file: main.tex
\documentclass{article}
\usepackage[]{arxiv}

\usepackage[utf8]{inputenc} 
\usepackage{hyperref}       
\usepackage{url}            
\usepackage{booktabs}       
\usepackage{amsfonts}       
\usepackage{nicefrac}       
\usepackage{microtype}      
\usepackage{amsmath,amssymb,amsfonts}
\usepackage{algorithmic}
\usepackage{mathrsfs}
\usepackage{graphicx}
\usepackage{textcomp}
\usepackage{xcolor}
\usepackage{yfonts}
\usepackage{cleveref}
\usepackage{dsfont}
\usepackage{amsthm}
\newtheorem{theorem}{Theorem}[]
\newtheorem{corollary}{Corollary}[]
\newtheorem{lemma}[]{Lemma}

\newtheorem*{assumption*}{Assumption}

\newtheorem{prop}{Proposition}
\newtheorem{definition}{Definition}

\usepackage{amsmath}

\DeclareMathOperator*{\argmin}{arg\,min}
\DeclareMathOperator*{\E}{\mathbb{E}}
\DeclareMathOperator*{\R}{\mathbb{R}}

\usepackage{mathtools}
\graphicspath{ {./images/} }
\usepackage{dsfont}
\usepackage{titling}
\usepackage[ruled,vlined]{algorithm2e}
\newtheorem*{lemma*}{Lemma} 
\predate{}
\postdate{}

\date{}
\title{Model-Agnostic Meta-Policy Optimization via Zeroth-Order Estimation: A Linear Quadratic Regulator Perspective }

\author{Yunian Pan, Tao Li, and Quanyan Zhu}

\begin{document}

\maketitle



\begin{abstract}
Meta-learning has been proposed as a promising machine learning topic in recent years, with important applications to image classification, robotics, computer games, and control systems. 
In this paper, we study the problem of using meta-learning to deal with uncertainty and heterogeneity in ergodic linear quadratic regulators. 
We integrate the zeroth-order optimization technique with a typical meta-learning method, proposing an algorithm that omits the estimation of policy Hessian, which applies to tasks of learning a set of heterogeneous but similar linear dynamic systems. 
The induced meta-objective function inherits important properties of the original cost function when the set of linear dynamic systems are meta-learnable, allowing the algorithm to optimize over a learnable landscape without projection onto the feasible set. 
We provide stability and convergence guarantees for the exact gradient descent process by analyzing the boundedness and local smoothness of the gradient for the meta-objective, which justify the proposed algorithm with gradient estimation error being small. We provide the sample complexity conditions for these theoretical guarantees,  as well as a numerical example at the end to corroborate this perspective.
\end{abstract}



\input{introduction}

\input{relatedwork}

\input{problem}

\input{methodology}

\input{analysis}

\input{conclusion}







\section*{ACKNOWLEDGMENT}

We gratefully acknowledge Leonardo F. Toso from Columbia University for his indispensable insights into the technical details of this work, and we thank Prof. Ba\c{s}ar for his invaluable discussions during the second author's visit to University of Illinois Urbana-Champaign.


\input{main.bbl}
\appendix
\input{appendix}

\end{document}

%% file: introduction.tex
\section{INTRODUCTION}


Recent advancements in meta-learning, a machine learning paradigm addressing the learning-to-learn challenge \cite{hochreiter01meta-recurrent}, have shown remarkable success across diverse domains, including robotics \cite{Zhao_2023,10161368}, image processing \cite{pan2023ordermetastackelbergmethod,li2024metastackelberggamerobust}, and cybersecurity \cite{10225816}.
One epitome of the various meta-learning approaches is Model-Agnostic Meta-Learning (MAML) \cite{finn2017model}. Compared with other deep-learning-based meta-learning approaches \cite{meta_survey}, MAML formulates meta-learning as a stochastic compositional optimization problem \cite{liu17composite,tianyi21composite}, aiming to learn an initialization that enables rapid adaptation to new tasks with just a few gradient updates computed using online samples. 

Since MAML is model-agnostic (compatible with any model trained with gradient descent), it is a widely applicable framework. In supervised learning (e.g., image recognition, speech processing), where labeled data is scarce, MAML facilitates few-shot learning \cite{song2023comprehensive}, enabling models to learn new tasks with minimal examples. In reinforcement learning (RL) (e.g., robotic control, game playing), MAML allows agents to generalize across multiple environments, leading to faster adaptation in dynamic and partially observable settings \cite{10161368,10225816}. Additionally, as a gradient-based optimization method, MAML benefits from its mathematical clarity, making it well-suited for theoretical analysis and highly flexible for further enhancements.

In the RL domain, MAML samples a batch of dynamic systems from an agnostic environment, i.e., a distribution of tasks, then optimizes the policy initialization with regard to the anticipated post-policy-gradient-adaptation performance, averaging over these tasks. The policy initialization will then be fine-tuned at test time. 
The complete MAML policy gradient methods for such a meta-objective require differentiating through the optimization process, which necessitates the estimation of Hessians or even higher order information, making them computationally expensive and unstable, especially when a large number of gradient updates are needed at test time \cite{fallah2021convergencetheorydebiasedmodelagnostic,nichol2018firstordermetalearningalgorithms, li2024metastackelberggamerobust}. This incentivizes us to focus our attention on the first-order implementation of MAML, unlike reptile \cite{nichol2018firstordermetalearningalgorithms}, which simply neglects the computation of Hessians or higher order information when estimating the gradient for meta-objective, we develop a framework that still approximates the exact gradient of the meta-objective, with controllable bias that benefits from the smoothness of the cost functional. This methodology stems from the zeroth-order methods, more specifically, Stein's Gaussian smoothing \cite{stein1972bound} technique.

We choose the Linear Quadratic Regulator (LQR) problem as a testbed for our analysis, as it is a fundamental component of optimal control theory. The Riccati equation, derived from the Hamilton-Jacobi equation \cite{bacsar1998dynamic}, provides the linear optimal control gain for LQR problems. While LQR problems are analytically solvable, they can still benefit from reinforcement learning (RL) and meta-RL, particularly in scenarios where model information is incomplete—a setting known as model-free control (see \cite{abbasi2011online,abbasi2011regret,pmlr-v97-cohen19b} for related works). Our focus is on the policy optimization of LQRs, specifically in refining an initial optimal control policy for a set of similar Linear Time-Invariant (LTI) systems, which share the same control and state space but differ in system dynamics and cost functionals. A practical example of such a scenario is a robotic arm performing a repetitive task, such as picking up and placing multiple block objects in a specific order. Each time the robot places a block, the system dynamics shift, requiring rapid adaptation to maintain optimal performance.


Our contribution is twofold. First, we develop a zeroth-order meta-gradient estimation framework, presented in Algorithm \ref{zerothordermethods}. This Hessian-free approach eliminates the instability and high computational cost associated with exact meta-gradient estimation. Second, we establish theoretical guarantees for our proposed algorithms. Specifically, we prove a stability result (\Cref{thm:stabilityattain}), ensuring that each iteration of Algorithm \ref{metalqr} produces a stable control policy initialization across a wide range of tasks. Additionally, we provide a convergence guarantee (\Cref{thm:convergence}), which ensures that the algorithm successfully finds a local minimum for the meta-objective. Our method is built on simultaneous perturbation stochastic approximation \cite{spall97spsa, flaxman2004onlineconvexoptimizationbandit} with a close inspection of factors influencing the zero-th order gradient estimation error, including the perturbation magnitude, roll-out length of sample trajectories, batch size of trajectories, and interdependency of estimation errors arising in inner gradient adaptation and outer meta-gradient update. We believe the developed technique in controlling the estimation error and associated high-probability error bounds would benefit the future work on biased meta-learning (in contrast to debiased meta-learning \cite{fallah2021convergencetheorydebiasedmodelagnostic}), which trades estimation bias for lesser computation complexity. Even though this work studies LQRs, our zero-th order policy optimization method easily lends itself to generic Makrov systems (e.g., \cite{li2024metastackelberggamerobust}) for efficient meta-learning algorithm design.

%% file: relatedwork.tex
\section{RELATED WORK}
\subsection{Policy Optimization (PO)}

Policy optimization (PO) methods date back to the 1970s with the model-based approach known as differential dynamic programming \cite{gershwin1970discrete}, which requires complete knowledge of system models. In model-free settings, where system matrices are unknown, various estimation techniques have emerged. Among these, finite-difference methods approximate the gradient by directly perturbing the policy parameters, while REINFORCE-type methods \cite{williams1992simple} estimate the gradient of the expected return using the log-likelihood ratio trick. For LQR tasks, however, analyzing the state-control correlations in REINFORCE-type methods poses significant challenges \cite{fazel2018global, DBLP:journals/corr/abs-2011-10300}. Therefore, we build our framework on finite-difference methods and develop a novel meta-gradient estimation procedure tailored specifically for the model-agnostic meta-learning problem. Overall, PO methods have been well established in the literature (see \cite{fazel2018global,malik2019derivative,gravell2020learning,hu2023toward}).

Zeroth-order methods have garnered increasing attention in policy optimization (PO), particularly in scenarios where explicit gradient computation is infeasible or computationally expensive. Rather than relying on REINFORCE-type methods for direct gradient evaluations, zeroth-order techniques estimate gradients using finite-difference methods or random search-based approaches. A foundational work in this domain is the Evolution Strategies (ES) method \cite{salimans2017evolution}, which reformulates PO as a black-box optimization problem, obtaining stochastic gradient estimates through perturbed policy rollouts. Similarly, \cite{allen2023scalablefinitedifferencemethod} introduces a method that leverages policy perturbation while efficiently utilizing past data, improving scalability. These approaches are particularly valuable in settings where Hessian-based computations or higher-order derivative information are impractical, driving the development of Hessian-free meta-policy optimization frameworks.

\subsection{Model-Agnostic Meta-Learning (MAML)}

The concept of meta-learning, or learning to learn, involves leveraging past experiences to develop a control policy that can efficiently adapt to novel environments, agents, or dynamics. One of the most prominent approaches in this area is MAML (Model-Agnostic Meta-Learning) as proposed by \cite{finn2017model,finn2019online}. MAML is an optimization-based method that addresses task diversity by learning a "common policy initialization" from a diverse task environment.
Due to its success across various domains in recent years, numerous efforts have been made to analyze its theoretical convergence properties. For instance, the model-agnostic meta-RL framework has been studied in the context of finite-horizon Markov decision processes by \cite{fallah2020provably,fallah2021convergencetheorydebiasedmodelagnostic,liu2022theoretical,beck2023survey}. However, these results do not directly transfer to the policy optimization (PO) setting for LQR, because key characteristics of the LQR cost objective—such as gradient dominance and local smoothness—do not straightforwardly extend to the meta-objective.

For example, \cite{molybog2020global} demonstrates that the global convergence of MAML over LQR tasks depends on a global property assumption ensuring that the meta-objective has a benign landscape. Similarly, \cite{musavi2023convergence} establishes convergence under the condition that all LQR tasks share the same system dynamics. It was not until \cite{toso2024meta} that comprehensive theoretical guarantees began to emerge: their analysis provided personalization guarantees for MAML in LQR settings by explicitly accounting for heterogeneity across different LQR tasks. The result readily passes the sanity check; the performance of the meta-policy initialization is affected by the diversity of the tasks.

All the aforementioned MAML approaches involve estimating second-order information, which can be problematic in LQR settings where the Hessians become high-dimensional tensors. Although recent studies such as \cite{toso2024meta,balasubramanian2019zerothordernonconvexstochasticoptimization} have employed advanced estimation schemes to mitigate these challenges, issues related to computational burden and numerical stability persist. Motivated by Reptile \cite{nichol2018firstordermetalearningalgorithms}, a first-order meta-learning method, we adopt a double-layered zero-th order meta-gradient estimation scheme that skips the Hessian tensor estimation. Our work extends the original work in \cite{pan2024modelagnosticzerothorderpolicyoptimization} by providing a comprehensive analysis of the induced first-order method, thereby offering a more computationally efficient and stable alternative for meta-learning in LQR tasks.


%% file: problem.tex
\section{PROBLEM FORMULATION}

\subsection{Preliminary: Policy Optimization for LQRs}

Let $\mathcal{T} = \{(A_i, B_i, Q_i, R_i)\}_{i \in [I]}$ be the finite set of LQR tasks, where $[I] := \{1, \ldots, I\}$ is the task index set, $A_i \in \R^{ d \times d}, B_i \in \R^{d \times k}$ are system dynamics matrices of the same dimensions, $Q_i \in \R^{ d \times d}, R_i \in \R^{k \times k}$, and $Q_i, R_i \succeq 0$ are the associated cost matrices. 
We assume a prior probability distribution $p \in \Delta(\mathcal{T})$ which we can sample the LQR tasks from.
For each LQR task $i$, the system is assumed to share the same state space $ \R^d$ and control space $ \R^k$, and is governed by the stochastic linear dynamics associated with some quadratic cost functions: 
\begin{equation*}
    x_{t+1} = A_i x_{t} + B_i u_{t} + w_t, \quad\quad g_i(x_t, u_t) = x_t^{\top}Q_i x_t + u_t^{\top} R_i u_t ,
\end{equation*}
where  $x_t \in \R^d$, $u_t \in \R^k$, $w_t$ are some random i.i.d. zero-mean noise with and covariance matrix $\Psi$, which is symmetric and positive definite. 

For each system $i$, our objective is to minimize the average infinite horizon cost, 
\begin{equation*}
    J_i = \lim_{T \to \infty}\frac{1}{T}\mathbb{E}_{x_0 \sim \rho_0, \{w_t\} }\left[\sum_{t=0}^{T-1}g_i(x_t, u_t)\right],
\end{equation*}
where $\rho_0$ is the initial state distribution $\mathcal{N}(0, \Sigma_0)$ with $\Sigma_0 \geq \mu I$ for some $\mu \geq 0$. 
For task $\mathcal{T}_i$, the optimal control $\{u_t^{i*}\}_{t \geq 0} $ can be expressed as $u^{i*}_t = - K^{*}_i x_t$, where $K^{*}_i \in \R^{k \times d}$ satisfies $K^{*}_i=\left(R_i+B_i^{\top} P_i^{*} B_i\right)^{-1} B_i^{\top} P^{*}_i A_i$, and $P^{i}_*$ is the unique solution to the following discrete algebraic Riccati equation 
$ P^{i}_*=Q_i+A_i^{\top} P^*_{i} A_i+A_i^{\top} P^{i}_*  B_i\left(R_i+B_i^{\top} P^{i}_* B_i\right)^{-1}\ B_i^{\top} P^{i}_* A_i$.

A policy $K \in \R^{d\times k}$ is called \textit{stable} for system $i$ if and only if $\rho(A_i - B_i K) < 1$, where $\rho(\cdot)$ stands for the spectrum radius of a matrix.
Denoted by $ \mathcal{K}_i$ the set of stable policy for system $i$, let $\mathcal{K} := \bigcap_{i \in [I]} \mathcal{K}_i$.
For a policy $K \in \mathcal{K}_i$, the induced cost over system $i$ is 
\begin{equation*}
\begin{aligned}
     J_i(K) & = \lim_{T \to \infty}\frac{1}{T}\mathbb{E}_{x_0 \sim \rho_0, w_t }[\sum_{t-0}^{T-1}\left(x_{t}^{\top}( Q_i + K^{\top}R_iK) x_{t}\right)] \\
      = \E_{x \sim \rho^i_K} &[x^{\top}( Q_i + K^{\top}R_iK) x]  = \operatorname{Tr}\left[(Q_i + K^{\top} R_i K) \Sigma^i_K \right],
\end{aligned}
\end{equation*}
where the limiting stationary distribution of $x_t$ is denoted by $\rho^i_K$, $\operatorname{Tr}(\cdot)$ stands for the trace operator. The Gramian matrix $\Sigma^i_K:= \mathbb{E}_{x \sim \rho_K^i} [xx^{\top}] = \lim_{T \to \infty} \E_{x_0 \sim \rho_0 } [\frac{1}{T} \sum_{t=0}^{T-1} x_t x_t^{\top}]$ satisfies the following Lyapunov equation
\begin{equation}\label{xgramian}
    \Sigma^i_{K}= \Psi +(A_i-B_i K) \Sigma^i_{K}(A_i-B_i K)^{\top}.
\end{equation}
\eqref{xgramian} can be easily verified through elementary algebra. 


\begin{prop}[Policy Gradient for LQR \cite{fazel2018global,DBLP:journals/corr/abs-1907-06246,bu2019lqrlensordermethods}] \label{prop1}
For any task $\mathcal{T}_i$, the expression for average cost is $J_i(K) = \operatorname{Tr}(P^i_K )$, 
and the expression of $\nabla J_i(K)$ is 
 \begin{equation}
   \begin{aligned}
       \nabla J_i(K)  &= 2\left[\left(R_i+B_i^{\top} P_{K} B_i\right) K-B_i^{\top} P^i_{K} A_i\right] \Sigma^i_{K}\\ &=2 E^i_{K} \Sigma^i_{K}
   \end{aligned}
 \end{equation}
 where $\Sigma^i_{K}$ satisfies \eqref{xgramian}, $E^i_K$ is defined to be
 \begin{equation*}
     E^i_K := \left(R_i+B_i^{\top} P^i_{K} B_i\right) K-B_i^{\top} P^i_{K} A_i,
 \end{equation*}
 and $P^i_K$ is the unique positive definite solution to the Lyapunov equation.
 \begin{equation*}
     P^i_{K}=\left(Q_i+K^{\top} R_i K\right)+(A_i-B_i K)^{\top} P^i_{K}(A_i - B_i K).
 \end{equation*}
 The Hessian operator $\nabla J_i(K)$ acting on some $X \in \R^{ k \times d}$ is given by, 
 \begin{equation}
     \nabla^{2} J_i(K)[X]:=2\left(R_i+B^{\top}_i P_{K}^i B_i\right) X \Sigma_{K}^i-4 B^{\top}_i\tilde{P}_{K}^i[X]\left(A_i-B_i K\right) \Sigma_{K}^i
 \end{equation}
 where $ \tilde{P}_{K}^i[X]$ is the solution to
 \begin{equation*}
      \tilde{P}_{K}^i[X]:=\left(A_i-B_i K\right)^{\top} \tilde{P}_{K}^i[X]\left(A_i-B_i K\right)+X^{\top} E_{K}^i+E_{K}^{(i) \top} X .
 \end{equation*}
\end{prop}

It is, therefore, possible to employ the first- and second-order algorithms to find the optimal controller for each specific task, in the model-based setting where the gradient/Hessian expressions are computable, see, e.g., in \cite{fazel2018global} for the following three first-order methods:
\begin{align*}
    K_{n+1} & =K_{n}-\eta \nabla J_i \left(K_{n}\right) \quad \quad & \text{Gradient Descent}\\
K_{n+1} & =K_{n}-\eta \nabla J_i \left(K_{n}\right) (\Sigma^i_{K_{n}})^{-1}  \quad   & \text{Natural Gradient Descent}  \\
K_{n+1}& =K_{n}-\eta\left(R_i+B_i^{\top} P^i_{K_{n}} B_i\right)^{-1} \nabla J_i\left(K_{n}\right) (\Sigma^i_{K_{n}})^{-1} \quad   & \text{Gauss-Newton} 
\end{align*}
 
Our discussion hitherto has focused on the deterministic policy gradient, where the policy is of linear form and depends on the policy gain $K$ deterministically. Yet, we remark that a common practice in numerical implementations is to add a Gaussian noise to the policy to encourage exploration, arriving at the linear-Gaussian policy class \cite{yang19pg-lqr}:
\begin{equation*}
    \{u_K(\cdot|x)=\mathcal{N}(-Kx, \sigma^2 I_k), K\in \mathbb{R}^{d\times k}\}. 
\end{equation*}
Such a stochastic policy class often relies on properly crafted \textit{regularization} for improved sample complexity and convergence rate \cite{kakakde21pg}. For stochastic policies, entropy-based regularization receives a significant amount of attention due to its empirical success \cite{ahmed19entropy}, of which softmax policy parametrization  \cite {mei20softmax, kakakde21pg} and entropy-based mirror descent \cite{pan-tao23noneq, pan-tao24delay, pan-tao24mirror-play} are well-received regularized policy gradient methods. We refer the reader to \cite[Sec. 2]{tao22confluence} for the connection between softmax and mirror descent methods. Finally, we remark that the policy gradient characterization in the stochastic case admits the same expression as in the deterministic counterpart. Hence, we limit our focus to the deterministic case to avoid additional discussion on the variance introduced by the stochastic policy.


\subsection{Meta-Policy-Optimization}

In analogy to \cite{finn2017model,fallah2020provably}, we consider meta-policy-optimization, which draws inspiration from Model-Agnostic-Meta-Learning (MAML) in the machine learning literature. 
Our objective is to find a meta-policy initialization, such that one step of (stochastic) policy gradient adaptation still attains optimized on-average performance for the tasks $\mathcal{T}$:
\begin{equation} \label{obj}
    \min_{K \in \bar{\mathcal{K}}} \mathcal{L}(K) := 
\mathbb{E}_{i \sim p}\left[J_{i}\left(\underbrace{ K  - \eta {\nabla} J_{i}\left(K\right)}_{\text{one-step adaptation}}\right) \right] ,
\end{equation}
where $\bar{\mathcal{K}}$ is the admissible set. At first glance, one might define $\bar{\mathcal{K}}$ as simply the intersection of all $\mathcal{K}_i$, however, this approach may render the problem ill-posed, since the functions $J_i(\cdot)$ can be ill-defined if the one-step gradient adaptation overshoots. Thus, with a given adaptation rate $\eta$, we define $\bar{\mathcal{K}}$ as in \Cref{mamlstabilizing}.

\begin{definition}[MAML-stablizing \cite{musavi2023convergence}]\label{mamlstabilizing}
 With a proper selection of adaptation rate $\eta$, a policy $K$ is MAML-stabilizing if for every task $i \in \mathcal{T}$, $\rho(A_i - B_i K ) <  1$ and  $ \rho(A_i - B(K - \eta \nabla J_i(K))) < 1$, we denote this set by $\bar{\mathcal{K}}$. 
\end{definition}

Definition \ref{mamlstabilizing} prepares us to adopt the first-order method to solve this problem, with learning iteration defined as follows:
\begin{equation*}
\begin{aligned}
     & \quad  K_{n+1} = K_n -\eta \nabla \mathcal{L} (K_n), \\ 
    & \text { where } \nabla \mathcal{L}(K):=\mathbb{E}_{i \sim p} \left[   (I - \eta  \nabla^2 J_i (K) )\nabla J_i \left(K^{\prime}\right) \right] \\
   & \quad \quad \quad  K^{\prime} = K-\eta \nabla J_i (K). 
\end{aligned}
\end{equation*}

In general, an arbitrary collection of LQRs is not necessarily meta-learnable using gradient-based optimization techniques, as one might not be able to find an admissible initialization of policy gain.
For instance, consider a two-system scalar case where $A_1 = 3, B_1 = 4$ and $A_2 = 1, B_2 = -1$. The policy evaluation requires an initialization $K$ to be stable for both system, which means $K \in (\frac{1}{2}, 1) \cap (-2, 0) = \emptyset$! This example illustrates that in regards to LQR cases, not all collections of LTIs are meta-learnable using MAML. 

Therefore, it is reasonable to assume that the systems exhibit a degree of similarity such that the set of tasks remains MAML-learnable. This assumption not only necessitates that the joint stabilizing sets are nonempty, i.e., $\bigcap_{i \in [I]} \mathcal{K}_i \neq \emptyset$, but also requires the existence of a set of MAML-stabilizing policies, $\bar{\mathcal{K}} \neq \emptyset$. We formalize such requirements in the definition below.

\begin{definition}[Stabilizing sub-level set \cite{toso2024meta}] \label{def:stabilizing_set} The task-specific and MAML stabilizing sub-level sets are defined as follows:

\begin{itemize}
    \item Given a task $\mathcal{T}_i$, the task-specific sub-level set $\mathcal{S}_i \subseteq \mathcal{K}_i $ is
    \begin{align*}
        \mathcal{S}_i := \left\{K\; | \; J_i(K) - J_i(K^\star_i) \leq \gamma_i \Delta^{i}_0\right\}, \text{ with } \Delta^{i}_0 : = J_i(K_0) - J_i(K^\star_i).  
    \end{align*}
    where $K_0$ denotes an initial control gain for the first-order method and $\gamma_i$ being any positive constant.

    \item The MAML stabilizing sub-level set $\mathcal{S} \subseteq \bar{\mathcal{K}}$ is defined as the intersection between each task-specific stabilizing sub-level set, i.e., ${\mathcal{S}} := \cap_{i\in [I]} \mathcal{S}_i$. 
\end{itemize}
\end{definition}

It is not hard to observe that, once $K \in \mathcal{S}$, it is possible to select a small adaptation rate $\eta$, such that $K^{\prime} \in \mathcal{S}$, in other words, $\eta$ controls whether $K \in \bar{\mathcal{K}}$. This property will be formalized later in section \ref{sec:gdanalysis}. For now, we simply assume that we have access to an admissible initial policy $K_0  \in \mathcal{S}$. Readers can refer to \cite{perdomo2021stabilizing} and \cite{ozaslan2022computing} for details on how to find an initial stabilizing controller for the single LQR instance.

%% file: methodology.tex
\section{METHODOLOGY}\label{sec:methodology}

\subsection{Zero-th Order Methods}

In the model-free setting where knowledge of system matrices is absent, sampling and approximation become necessary.
In this case, one can sample roll-out trajectories, from the specific task $i$ to perform the policy evaluation from $K$, then, optimize the system performance index through policy iteration.

The zeroth-order methods are derivative-free optimization techniques that allow us to optimize an unknown smooth function $J_i(\cdot): \R^{k \times d} \to \R$ by estimating the first-order information \cite{flaxman2004onlineconvexoptimizationbandit,spall97spsa}. What it requires is to query the function values $J_i$ at some input points. 
A generic procedure is to firstly sample some perturbations $U \sim \operatorname{Unif}(\mathbb{S}_r)$, where $\mathbb{S}_r := \{ r \in \mathbb{R}^{k\times d} \big\vert  \|r\|_F = r \}$ is a $r$-radius $k \times d$-dimensional sphere, and estimate the gradient of the perturbed function through equation:
\begin{equation}\label{steinidtty}
    \nabla_r J_i(K) = \frac{ dk }{ r^2} \E_{U \sim \operatorname{Unif}(\mathbb{S}_r)}[ J_i(K+U)U].
\end{equation}
Based on Stein's identity \cite{stein1972bound} and Lemma 2.1 \cite{flaxman2004onlineconvexoptimizationbandit}, $ \mathbb{E}[\nabla J_i (K+U)] = \nabla_r J_i(K)$, hence we obtain a perturbed version of the first-order information. 
The expectation $\E_{U \sim \operatorname{Unif}(\mathbb{S}_r)}$ can be evaluated through Monte-Carlo sampling. 
However, as we discussed, a function value oracle, i.e., the value of $J_i$ is not always accessible. 
One can substitute $J_i$ with the return estimates obtained from sample roll-outs, as demonstrated in Algorithm \ref{zerothorderalgo}, (adapted from \cite{fazel2018global}.) This type of gradient-estimation procedure samples trajectories with a perturbed policy $K +U$, instead of the target policy $K$. 

\begin{algorithm}[htbp]
\label{zerothorderalgo}
\SetKwInOut{Input}{Input}
\SetAlgoLined
\Input{Task simulator $i$, Policy $K$, number of trajectories $M$, \\ roll out length $\ell$, smoothing parameter $r$.}
\For{$m = 1, 2, \ldots, M$}{
     Sample a perturbed policy $K + U_m$, where $U_m$ is drawn uniformly from $\mathbb{S}_r$\;
     Simulate $K+U_m$ for $\ell$ steps starting from $x_0 \sim \rho_0$. Let $\widehat{J}^{(\ell)}_i(K+U_m)$ and $\tilde{\Sigma}^{i, (\ell)}_{K+U_m}$ be empirical estimates:
     \begin{align*}
         \tilde{J}^{(\ell)}_{i}(K+ U_m) & = \frac{1}{\ell}\sum_{l=1}^{\ell} g_{i}(x_l, - (K+U_m )x_l), \\
         \quad \quad \tilde{\Sigma}^{i, (\ell)}_{K+U_m} & = \frac{1}{\ell}\sum_{l=1}^{\ell} x_{l} x_{l}^{\top} , 
     \end{align*}
     
     where $g_t$ and $x_t$ are costs and states of the current trajectory $m$.
}
Return the (biased) estimates:
\begin{equation*}
     \tilde{\nabla} J_i(K) = \frac{1}{M} \sum_{m=1}^M \frac{dk}{ r^2} \tilde{J}^{(\ell)}_{i}(K+ U_m) U_m,  
 \end{equation*}
\caption{\texttt{Gradient Estimation} \cite{fazel2018global}}
\end{algorithm}

Algorithm \ref{zerothorderalgo} enables us to perform inexact gradient iterations such as $K^{\prime} = K - \eta \tilde{\nabla}J_i(K)$, where $\eta$ is the adaptation rate. 
However, there are two issues that persist. First, one has to restrict $r$ to be small so that the change on $K$ is not drastic, and the perturbed policy is admissible $K+U \in \mathcal{K}_i$. (We will provide theoretical guarantees later.)  Second,  the first-order optimization requires that the updated policy $K^{\prime}$ must be stable as well, even if the perturbed policy is stable, it is questionable how small the smoothing parameter $r$ and the adaptation rate $\eta$ should be to prevent the updated policy $K^{\prime}$ from escaping the admissible set. As has been demonstrated in \cite{fazel2018global}, the remedy to this is that when the cost function is locally smooth, it suffices to identify the regime of such smoothness and constrain the gradient steps within such regime.

Even though a single LQR task objective becomes infinite as soon as $A_i - B_i K$ becomes unstable, as established in \cite{fazel2018global} as well as in non-convex optimization literature, the (local) smoothness and gradient domination properties almost immediately imply global convergence for the gradient descent dynamics, with a linear convergence rate. We now hash out three core auxiliary results that lead to such properties. These results can be found in \cite{fazel2018global,wang2023fedsysid,bu2019lqrlensordermethods,musavi2023convergence}, we defer the explicit definition of the parameters to the appendix.

\begin{lemma}[Uniform bounds \cite{toso2024meta}]\label{lem:uniformbounds}
Given a LQR task $\mathcal{T}_i$ and an stabilizing controller $K \in \mathcal{S}$, the Frobenius norm of gradient $\nabla J_i(K)$, Hessian $\nabla^2 J_i(K)$ and control gain $K$ can be bounded as follows:
$$
\|\nabla J_i(K)\|_F \leq h_G(K), \text{ } \|\nabla^2 J_i(K)\|_F \leq h_H(K), \text { and } \|K\|_F \leq h_c(K), 
$$
where $h_G, h_H,$ and $h_c$ are problem dependent parameters.
 \end{lemma}

\begin{lemma}[Perturbation Analysis \cite{toso2024meta,musavi2023convergence}]\label{lem:perturbationanalysis}
    Let $K, K^{\prime} \in \mathcal{S} $ such that $\|\Delta\| := \|K^{\prime} -K\|  \leq h_{\Delta}(K)\ <\infty$, then, we have the following set of local smoothness properties:
    \begin{equation*}
        \begin{aligned}
             &\left|J_i\left(K^{\prime}\right)-J_i(K)\right| \leq h_{\text {cost}}(K) J_i(K) \|\Delta\|_F, \\
&\left\|\nabla J_i\left(K^{\prime}\right)-\nabla J_i(K)\right\|_F \leq h_{\text {grad}}(K)\|\Delta\|_F, \\
&\left\|\nabla^2 J_i\left(K^{\prime}\right)-\nabla^2 J_i(K)\right\|_F \leq h_{\text {hess}}(K)\|\Delta\|_F,
        \end{aligned}
    \end{equation*}
    for all tasks $ i \in [I]$,
    where $h_{\text {cost}}(K), h_{\text {grad}}(K), h_{\text {hess}}(K)$ are problem-dependent parameters. 
    
\end{lemma}

\begin{lemma}[Gradient Domination \cite{fazel2018global,yang19pg-lqr}]\label{lem:graddom}
For any LQR task $i \in [I]$, let $K^*_i$ be the optimal policy. Suppose $K \in \mathcal{S}$ has finite
cost. Then, it holds that 
\begin{equation*}
   \begin{aligned}
   J_i(K)-J_i\left(K^{*}_i\right)  & \geq \mu \cdot\frac{\operatorname{Tr} \left(E_{K}^{i,\top} E^i_{K}\right)}{\left\|R_i+B_i^{\top} P^i_{K} B_i\right\| } , \\ J_i(K)-J_i\left(K^{*}_i\right) & \leq \frac{1} { \sigma_{\min }(R_i)} \cdot\left\|\Sigma^i_{K^{*}}\right\| \cdot \operatorname{Tr}\left(E_{K}^{i,\top} E^i_{K}\right) \\
 & \leq \frac{\left\|\Sigma^i_{K^{*}}\right\| }{ \mu^2  \sigma_{\min}(R_i)} \|\nabla J_i (K)\|^2_F = : \frac{1}{\lambda_i } \|\nabla J_i (K)\|^2_F .
   \end{aligned}
\end{equation*}
\end{lemma}

\subsection{Hessian-Free Meta-Gradient Estimation}

Now we recall \eqref{steinidtty} and extend the zeroth-order technique to the meta-learning problem. Specifically, for problem \eqref{obj}, we derive a gradient expression for the perturbed objective function $\mathcal{L}$, thereby eliminating the need to compute the Hessian.
\begin{equation*}
    \nabla_r \mathcal{L} (K) = \frac{dk}{r^2} \E_{i \sim p, U \sim \mathbb{S}_r} \left[J_i (K + U - \eta \nabla J_i (K+ U) ) U\right] .
\end{equation*}
To evaluate expectation $\E_{U \sim \mathbb{S}_r, i \sim p}$ we sample $M$ independent perturbation $U_m$ and a batch of tasks $\mathcal{T}_n$, then average the samples. 
To evaluate return $J_i(K+U - \eta \nabla J_i(K+U))$ we first apply algorithm \ref{zerothorderalgo} to obtain approximate gradient $\tilde{\nabla}J_i(K+U)$ for a single perturbed policy, then sample roll-out trajectories using the one-step updated policy $K+U - \eta \tilde{\nabla}J_i (K+U)$ to estimate its associated return. 

A comprehensive description of the procedure is shown in Algorithm \ref{zerothordermethods}. Essentially we aim to collect $M$ samples for return by perturbed policy ${K}_m^i$, which requires the original perturbed policy $\widehat{K}_m$ and the gradient estimate of it. To do so, we use Algorithm \ref{zerothorderalgo} as an inner loop procedure. After computing ${K}_m^i$ we simulate it for $\ell$ steps to get the empirical estimate of return $J_i(K+U_m - \eta \nabla J_i (K+U_m))$. The entire procedure of meta-policy-optimization is shown in Algorithm \ref{metalqr}.

\begin{algorithm}[htbp]
\label{zerothordermethods}
\SetKwInOut{Input}{Input}
\SetAlgoLined
\Input{Meta-environment $p$, policy $K$, number of perturbations $M$, learning rate $\eta$, roll-out length $\ell$, parameter $r$; }
Randomly draw systems batch $\mathcal{T}_n $ from meta-environment $p$\;
\For{all $i \in \mathcal{T}_n$}{
     \For{$m = 1, 2, \ldots, M$}{
      Sample a policy $\widehat{K}_m = K + U_m$, where $U_m$ is drawn uniformly from $\mathbb{S}_r$\;
      Estimate $\tilde{\nabla}J_i(\widehat{K}_m) \leftarrow \texttt{Gradient Estimation}(i, \widehat{K}_m, M, \ell, r)$\;
      Perform one-step gradient adaptation:
      \begin{equation} \label{onestepinalgo}
          {K}_m^i = \widehat{K}_m - \eta\tilde{\nabla}J_i(\widehat{K}_m);
      \end{equation}
      
      Estimate $\tilde{J}^{(\ell)}_{i}(K^i_m)$ from simulating $K^i_m$ for $\ell$ steps starting with $x_0 \sim \rho_0$:
      $$
       \tilde{J}^{(\ell)}_{i}(K^i_m) = \frac{1}{\ell} \sum_{t=1}^{\ell} g_i(x_l, -K^i_m x_l) .
      $$
     
     }
}
The meta-gradient estimation:
     \begin{equation*} 
         \tilde{\nabla} \mathcal{L}(K) = \frac{1}{|\mathcal{T}_n|}\sum_{i \in \mathcal{T}_n} \frac{1}{M} \sum_{m = 1}^M \frac{dk}{r^2}\tilde{J}^{(\ell)}_{i}(K^i_m) U_m
     \end{equation*}
\caption{\texttt{Meta-Gradient Estimation}}
\end{algorithm}

Further, we can easily extend the results in \Cref{lem:uniformbounds}, \Cref{lem:perturbationanalysis} to the meta-objective, to show the boundedness and Lipschitz properties of $\mathcal{L}(K), \nabla \mathcal{L}(K)$, as in \Cref{lem:uniformboundsl} and \Cref{lem:perturbationl}, whose proofs--which we defer to the Appendix \ref{sec:appa}--are straightforward given the previous characterizations. These results provide an initial sanity check for the first-order iterative algorithm.

\begin{lemma}\label{lem:uniformboundsl}
Given a prior $p$ over LQR task set $\mathcal{T}$, adaptation rate $\eta$, and an MAML stabilizing controller $K \in \mathcal{S}$, the Frobenius norm of gradient $\nabla \mathcal{L}(K)$ and control gain $K$ can be bounded as follows:
\begin{equation}
\begin{aligned}
       \| \nabla \mathcal{L}(K)\|_F & \leq h_{G, \mathcal{L}} (K) , 
\end{aligned}
\end{equation}
where $h_{G, \mathcal{L}}:= ( k + \eta h_H (K)) (1 + \eta h_{grad} (K)) h_G (K)  $ is dependent on the problem parameters.
\end{lemma}

\begin{lemma}[Perturbation analysis of $\nabla \mathcal{L}(K)$]\label{lem:perturbationl}
Let $K, K^{\prime} \in \mathcal{S} $ such that $\|\Delta\| := \|K^{\prime} -K\|  \leq h_{\Delta}(K)\ <\infty$, then, we have the following set of local smoothness properties,
 \begin{align*} 
  &  |\mathcal{L}(K^{\prime}) - \mathcal{L}(K)| \leq h_{\mathcal{L},cost} \|\Delta\|_F\\
   &   \| \nabla \mathcal{L} (K) - \nabla \mathcal{L} (K^{\prime})\|_F   \leq h_{\mathcal{L},grad} \| \Delta\|_F,
 \end{align*}
 where $h_{\mathcal{L}, cost}: = h_{cost} (1+ \eta h_{grad}(K)) $ and $h_{\mathcal{L}, grad} := \eta h_{hess}(K) ( 1 + \eta h_{grad}) h_G (K) + (k + \eta h_H(K^{\prime})) h_{hess}(K) (1 + \eta h_{hess}(K)) $ are problem dependent parameters. 
\end{lemma}

\begin{algorithm}[htbp]
\label{metalqr}
\SetKwInOut{Input}{Input}
\SetAlgoLined
\Input{Task prior $p$, number of perturbations $M$, adaptation rate $\eta$,  \\learning rate $\alpha$, roll-out length $\ell$, parameter $r$, tolerance $\varepsilon$;}
initialize feasible policy $K_0 \in \mathcal{S}$ \; 
\While{$ \| \tilde{\nabla}\mathcal{L}(K) \leq \varepsilon \|$}{
    $\tilde{\nabla} \mathcal{L}(K) \leftarrow $ \texttt{Meta-Gradient Estimation}$(p, K_{n}, M, \eta, \ell, r)$ \;
    Update policy:
    \begin{equation*}
           K_{n+1} =  K_n  - \alpha \tilde{\nabla} \mathcal{L}(K_n). 
    \end{equation*}
}
\caption{Model-Agnostic Meta-Policy-Optimization}
\end{algorithm}

%% file: analysis.tex
\section{GRADIENT DESCENT ANALYSIS}
\label{sec:gdanalysis}
Our theoretical analysis for Algorithm \ref{metalqr} can be divided into two primary objectives: stability and convergence. For stability, we demonstrate that by selecting appropriate algorithm parameters, every iteration $n$ of gradient descent satisfies $K_{n} \in \bar{\mathcal{K}}$, ensuring that both $K_{n+1}$ and $K_n$ remain in  $\mathcal{S}$; Regarding convergence, we establish that the learned meta-policy initialization eventually approximates the optimal policies for each specific task, and we provide a quantitative measure of this closeness.

\subsection{Controlling Estimation Error}

In the following, we present our results that characterize the conditions on the step-sizes  $\eta, \alpha$ and zeroth-order estimation parameters $M$, $\ell$, $r$, and task batch size $|\mathcal{T}_n|$,  for controlling gradient and meta-gradient estimation errors. The proofs are deferred to the appendix.
Overall, our observations are as follows:
\begin{itemize}
    \item The smoothing radius is dictated by the smoothness of the LQR cost and its gradient, as well as the size of the locally smooth set.
    \item The roll-out length is determined by the smoothness of the cost function and the level of system noise.
    \item The number of sample trajectories and sample tasks is influenced by a broader set of parameters that govern the magnitudes and variances of the gradient estimates.
    \item Inner loop estimation errors can propagate readily, particularly when the scale of the sample tasks is large.
\end{itemize}

\begin{lemma}[Gradient Estimation]\label{lem:gradestimate}
    For sufficiently small numbers $\epsilon, \delta \in (0,1)$, given a control policy $K$, let $\ell$, radius $r$, number of trajectories $M$ satisfying the following dependence, 
    \begin{align*}
        \ell & \geq h^1_{\ell} (\frac{1}{\epsilon}, \delta) : = \max\{ h_{\ell, grad}(\frac{1}{\epsilon}), h_{\ell, var}(\frac{1}{\epsilon}, \delta)\} \\
        r & \leq  h^1_r (\frac{1}{\epsilon}) := \min \{ 1/\bar{h}_{cost}, \underline{h}_{\Delta}, \frac{\epsilon}{4\bar{h}_{grad}}\} \\
        M & \geq h^1_M (\frac{1}{\epsilon}, \delta) := h_{sample}(\frac{4}{\epsilon}, \delta)
    \end{align*}
    Then, with probability at least $1 - 2\delta$, the gradient estimation error is bounded by
    \begin{equation}
        \| {\nabla} J_i (K)- \tilde{\nabla} J_i (K) \|_F \leq \epsilon, 
    \end{equation}
     for any task $i \in [I]$.
\end{lemma}

\begin{lemma}[Meta-gradient Estimation]\label{lem:metagradestimate}
    For sufficiently small numbers $\epsilon, \delta \in (0,1)$, given a control policy $K$, let $\ell$, radius $r$, number of trajectory $M$ satisfies that 
    \begin{align*}
      | \mathcal{T}_n |& \geq  h_{sample, task} (\frac{2}{\epsilon}, \frac{\delta}{2}) , \\
        \ell & \geq  \max\{h^1_{\ell} (\frac{1}{\epsilon^{\prime}}, \delta^{\prime}) ,  h^2_{\ell, grad} (\frac{12}{\epsilon}), h^2_{\ell, var}(\frac{12}{\epsilon}, \delta^{\prime})\} , \\
        r & \leq \min \{ h^2_r(\frac{6}{\epsilon}),  h^1_r (\frac{1}{\epsilon}) \}  , \\
        M & \geq \max\{ h^2_M (\frac{1}{\epsilon}, \delta), h^1_M (\frac{1}{\epsilon^{''}}, \frac{\delta}{4})   \} , 
    \end{align*}
where $ h^2_M (\frac{1}{\epsilon}, \delta):= h_{sample}(\frac{1}{\epsilon^{''}}, \frac{\delta^{\prime}}{4})$, $\delta^{\prime} = \delta / h_{sample, task} (\frac{2}{\epsilon}, \frac{\delta}{2} ) $, and $\epsilon^{\prime} = \frac{\epsilon}{6 \frac{dk}{r} h_{cost} \bar{J}_{max}}$.
 Then, for each iteration the meta-gradient estimation is $\epsilon$-accurate, i.e., 
 \begin{align*}
     \|\tilde{\nabla} \mathcal{L} (K) - \nabla \mathcal{L} (K) \|_F \leq \epsilon
 \end{align*} 
 with probability at least $1 - \delta$. 
\end{lemma}




\subsection{Theoretical Guarantee}

We first provide the conditions on the step-sizes $\eta, \alpha$ and zeroth-order estimation parameters $M$, $\ell$, $r$, and $|\mathcal{T}_n|$, such that we can ensure that Algorithm \ref{metalqr} generates stable policies at each iteration. This stability result is shown in \Cref{thm:stabilityattain}. 

\begin{theorem} \label{thm:stabilityattain}
Given an initial stabilizing controller $K_0 \in \mathcal{S}$ and scalar $\delta \in (0,1)$,  let $\varepsilon_i := \frac{\lambda_i \Delta^i_0}{6}$, the adaptation rate $\eta \leq \min\{ \sqrt{\frac{1}{4(\bar{h}_{grad}^2 k^2 + \bar{h}_{grad}^2 \bar{h}_{H}^2 + \bar{h}_H^2 )}} , \frac{1}{4\bar{h}_{\text{grad}}}\}$, 
and $\varepsilon :=  \frac{\bar{\lambda}_i  \bar{\Delta}^i_0 (1 -2\phi_1)\phi_2 }{2 (1 + 4\phi_2 - 2\phi_1)}$ where $\phi_1 := 2(k^2 + \eta^2\bar{h}_H^2 )\eta^2 \bar{h}^2_{grad} + 2 \eta^2 \bar{h}_H^2$ and $\phi_2 := k^2 + \eta^2 \bar{h}_H^2) (2 + 2 \bar{h}^2_{grad} \eta^2$; let the learning rate $\alpha \leq \frac{\frac{1}{2} - \phi_1}{2 \phi_2 \bar{h}_{grad}}$. 
In addition, let the task batch size $|\mathcal{T}_n|$, the smoothing radius $r$, roll-out length $\ell$, and the number of sample trajectories satisfy:
\begin{align*}
      | \mathcal{T}_n |& \geq  h_{sample, task} (\frac{2}{\varepsilon}, \frac{\delta}{2}) , \\
        \ell & \geq  \max\{h^1_{\ell} (\frac{1}{\varepsilon_i}, \frac{\delta}{2}), h^1_{\ell} (\frac{1}{\varepsilon^{\prime}}, \delta^{\prime}) ,  h^2_{\ell, grad} (\frac{12}{\varepsilon}), h^2_{\ell, var}(\frac{12}{\varepsilon}, \delta^{\prime})\} , \\
        r & \leq \min \{ h^1_r (\frac{1}{\varepsilon_i}),  h^1_r (\frac{1}{\varepsilon}), h^2_r(\frac{6}{\varepsilon}) \}  , \\
        M & \geq \max\{ h^1_M(\frac{1}{\varepsilon_i}, \frac{\delta}{2}), h^1_M (\frac{1}{\varepsilon^{''}},\frac{\delta}{4})  h^2_M (\frac{1}{\varepsilon}, \delta)  \} , 
    \end{align*}
where $ h^2_M (\frac{1}{\varepsilon}, \delta):= h_{sample}(\frac{1}{\varepsilon^{''}}, \frac{\delta^{\prime}}{4})$, $\delta^{\prime} = \delta / h_{sample, task} (\frac{2}{\varepsilon}, \frac{\delta}{2} ) $, $\varepsilon^{\prime} = \frac{\varepsilon}{6 \frac{dk}{r} h_{cost} \bar{J}_{max}}$, $\varepsilon^{''} = \frac{\varepsilon}{6}$. 
Then, with probability at least $1-\delta$, Algorithm \ref{metalqr} yields a MAML stabilizing controller $K_n$ for every iteration, i.e., $K^i_n, K_n \in \mathcal{S}$, for all $n \in \{0,1, \ldots, N\}$, where
$K^i_n =  K_n - \eta \tilde{\nabla} J_i(K_n)$ is the updated policy for specific tasks $i \in [I]$.
\end{theorem}

The proof of stability result indicates that the learned MAML-LQR controller $K_N$ is sufficiently close to each task-specific optimal controller $K^\star_i$. The closeness of $K_N$ and $K^*_i$ can be measured by $J_i(K_N) - J_i(K^*_i)$, and because it is monotonically decreasing, we obtain stability for every iteration.

We proceed to give another set of conditions on the learning parameters, which ensure that the learned meta-policy initialization $K_N$ is sufficiently close to the optimal MAML policy-initialization $K^\star :=  \argmin_{K \in \bar{\mathcal{K}}} \mathcal{L} (K)$. For this purpose, we study the difference term $\mathcal{L} (K_N) - \mathcal{L}(K^\star)$.

\begin{theorem}\label{thm:convergence} (Convergence) 
Given an initial stabilizing controller $K_0 \in \mathcal{S}$ and scalar $\delta \in (0,1)$, let the parameters for Algorithm \ref{metalqr} satisfy the conditions in \Cref{thm:stabilityattain}. If, in addition, \begin{align*}
      | \mathcal{T}_n |& \geq  h_{sample, task} (\frac{2}{\bar{\varepsilon}}, \frac{\delta}{2}) , \\
        \ell & \geq  \max\{h^1_{\ell} (\frac{1}{\bar{\varepsilon}^{\prime}}, \delta^{\prime}) ,  h^2_{\ell, grad} (\frac{12}{\bar{\varepsilon}}), h^2_{\ell, var}(\frac{12}{\bar{\varepsilon}}, \delta^{\prime})\} , \\
        r & \leq \min \{ h^2_r(\frac{6}{\bar{\varepsilon}}),  h^1_r (\frac{1}{\bar{\varepsilon}}) \}  , \\
        M & \geq \max\{ h^2_M (\frac{1}{\bar{\varepsilon}}, \delta), h^1_M (\frac{1}{\bar{\varepsilon}^{''}}, \frac{\delta}{4})   \} , 
    \end{align*}
 where $\bar{\varepsilon} := \frac{\bar{\lambda}_i (1 - \eta^2 \bar{h}_H^2)\psi_0}{6}$, $\psi_0 :=  \mathcal{L} (K_0 ) - \mathcal{L} (K^\star)$, $ h^2_M (\frac{1}{\bar{\varepsilon}}, \delta):= h_{sample}(\frac{1}{\bar{\varepsilon}^{''}}, \frac{\delta^{\prime}}{4})$, $\delta^{\prime} = \delta / h_{sample, task} (\frac{2}{\bar{\varepsilon}}, \frac{\delta}{2} ) $, $\bar{\varepsilon}^{\prime} = \frac{\varepsilon}{6 \frac{dk}{r} h_{cost} \bar{J}_{max}}$, $\bar{\varepsilon}^{''} = \frac{\bar{\varepsilon}}{6}$,
Then, when $N \geq \frac{8}{\alpha \bar{\lambda}_i (1 - \eta^2 \bar{h}_H^2)}\log( \frac{2\psi_0}{\epsilon_0}) $, with probability $1- \bar{\delta}$, it holds that, 
\begin{align*}\label{eq:convergence_model_free}
    &\mathcal{L} (K_{N}) - \mathcal{L}(K^\star)  \leq \epsilon_0.
\end{align*}
\end{theorem}

\section{NUMERICAL RESULTS}
We consider three cases of state and control dimensions in the numerical example, but due to computational limits, we consider a moderate system collection size $I = 5$.
The collection of systems is randomly generated to behave ``similarly'', in the sense that the stabilizing sublevel set is admissible for some given initial controller. 
Specifically, we sample matrices $A_0, B_0, Q_0, R_0, \Psi_0$ from uniform distributions, and adjust $A_0$ so that $\rho(A_0) < 1$, adjust $Q_0, R_0, \Psi_0$ to be symmetric and positive definite. 
Then, we sample the rest of systems $i$ independently such that their system matrices are centered around $A_0, B_0, Q_0, R_0, \Psi_0$, (for example $[A_i]_{m,n} \sim \mathcal{N}([A_0]_{m,n}, 0.25)$ for some $i$, $m$ and $n$.) and follow the same procedure to make $\rho(A_i) < 1$ and $Q_i, R_i, \Psi_i$ positive definite.

\begin{figure}[htbp]
    \centering
    \includegraphics[width = .8\textwidth]{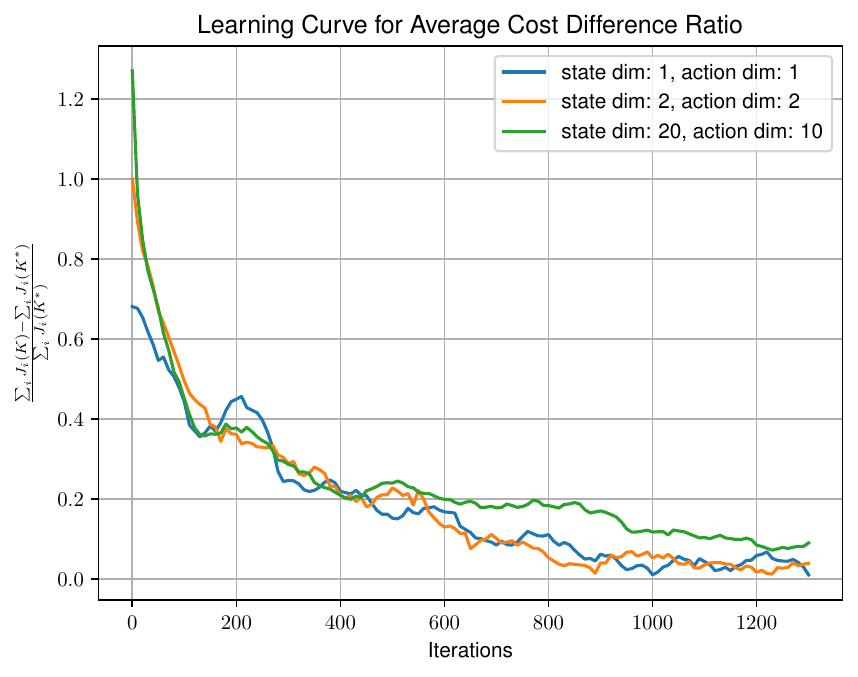}
    \caption{The plot shows three curves encapsulating the changing of average performance during gradient descent, each corresponds to a particular dimension setting of state and action space, (green: $d = 20, k = 10$, orange: $d = 2, k = 2$, blue: $d= 1, k =1$.) constant learning rates $\alpha = 1e-3$, $\eta = 1e-5$ for orange and blue cases and $\alpha = 1e-5$, $\eta = 1e-7$ for green curve, numbers of meta and inner perturbation $ M= 100$, gradient smooth parameter $r = 0.05$, roll out length $\ell = 50$.}
    \label{plot}
\end{figure}

We report the learning curves for average cost difference ratio $\frac{\sum_{i \in [I]} J_i(K_n) - J_i (K^{*}_i)}{\sum_{i \in [I]} J_i(K^{*}_i) }$, this quantity captures the performance difference between a one-fits-all policy and the optimal policy in an average sense. Fig. \ref{plot}. demonstrates the evolution of this quantity during learning for three cases. Overall, despite that there are oscillations due to the randomness of meta-gradient estimators, the ratios become sufficiently small after adequate iterations, which implies the effectiveness of the algorithm.

%% file: conclusion.tex
\section{CONCLUSIONS}
In this paper, we investigate a zeroth-order meta-policy optimization approach for model-agnostic LQRs. Drawing inspiration from MAML, we formulate the objective \eqref{obj} with the goal of refining a policy that achieves strong performance across a set of LQR problems using direct gradient methods.
Our proposed method bypasses the estimation of the policy Hessian, mitigating potential issues of instability and high variance. We analyze the conditions for meta-learnability and establish finite-time convergence guarantees for the proposed algorithm.
To empirically assess its effectiveness, we present numerical experiments demonstrating promising performance under the average cost difference ratio metric.
A promising direction for future research is to derive sharper bounds on the iteration and sample complexity of the proposed approach and explore potential improvements.

%% file: appendix.tex
\section*{APPENDIX}

In the following, we present the formal proofs and technical details supporting our main findings. To achieve this, we first give the elementary proof for the gradient and Hessian expression of the LQR cost. 

\label{appendix}
\proof[Proof of Prop. \ref{prop1}]{
 For arbitrary system $i$, consider a stable policy $K$ such that $\rho(A_i-B_iK) < 1$, define operator $\mathcal{T}_K(\Sigma)$ by:
 \begin{equation*}
      \mathcal{T}^i_K(\Sigma) = \sum_{t\geq 0} (A_i-B_i K)^t \Sigma [(A_i-B_i K)^t]^{\top} .
 \end{equation*}
 Here, $\mathcal{T}^i_K$ is an adjoint operator, observing that for any two symmetric positive definite matrices $\Sigma_1$ and $\Sigma_2$, we have
 \begin{equation*}
    \begin{aligned}
         \operatorname{Tr}( \Sigma_1 \mathcal{T}^i_K(\Sigma_2)) & =  \operatorname{Tr}(\sum_{t\geq 0} \Sigma_1 (A_i-B_i K)^t \Sigma_2 [(A_i-B_i K)^t]^{\top}) \\
         & = \operatorname{Tr} (\sum_{t\geq 0} [(A_i-B_i K)^t]^{\top} \Sigma_1 (A_i-B_i K)^t \Sigma_2 )\\ & = \operatorname{Tr}(\mathcal{T}^{i\top}_K(\Sigma_1) \Sigma_2)
    \end{aligned}
 \end{equation*}
 Meanwhile, since we know that $\Sigma_K^i$ satisfies recursion \eqref{xgramian},
$\Sigma^i_K = \mathcal{T}_K^i(\Psi) .$
Thus the average cost of $K$ for system $i$ can be written as
\begin{equation*}
\begin{aligned}
         J_i(K) &=  \operatorname{Tr}\left[\left(Q_i+K^{\top} R_i K\right) \cdot \Sigma^i_{K}\right] \\ &=\operatorname{Tr}\left[\left(Q_i + K^{\top} R_i K\right) \cdot \mathcal{T}_{K}^i\left(\Psi \right)\right]  \\ & =\operatorname{Tr}\left[\mathcal{T}_{K}^{i\top}\left(Q_i+K^{\top} R_i K\right) \cdot \Psi \right]=\operatorname{Tr}\left(P^i_{K} \Psi\right).
\end{aligned}
\end{equation*}
By rule of product:
\begin{equation*}
\begin{aligned}
  \nabla J_i(K) & = 2R_i K\Sigma_K^i + \nabla \operatorname{Tr}(Q_i^{\prime} \mathcal{T}_K^i(\Psi))|_{Q^{\prime} = Q_i + K^{\top}R_i K} 
\end{aligned}
\end{equation*}
Here, we derive the expression for the second term. 
For symmetric positive definite matrix $\Sigma$, define operator $\Gamma^i_K(\Sigma) := (A_i - B_i K) \Sigma (A_i - B_i K)^{\top}$, we have 
\begin{equation*}
  Q^{\prime}_i \mathcal{T}^i_K (\Sigma^i_K) = Q^{\prime}_i\Psi +  \Gamma^i_K ( \mathcal{T}_K^i(\Sigma_K^i)),
\end{equation*}
and 
$\mathcal{T}^i_K  (\Sigma)= \sum_{t = 0}^{\infty} (\Gamma^i_K)^t(\Sigma).$
Since $\mathcal{T}_K^i$ is linear and adjoint
$$
 \operatorname{Tr}(Q_i^{\prime} \mathcal{T}_K^i(\Psi)) = \operatorname{Tr}(Q_i^{\prime} \Psi) + \operatorname{Tr}( \Gamma^{i\top}_K(Q^{\prime}_i)\mathcal{T}_K^{i\top}(\Psi)).
$$
Take derivative on both sides and unfold the right-hand side:
\begin{equation*}
\begin{aligned}
    \nabla \operatorname{Tr}(Q_i^{\prime} \mathcal{T}_K^i &(\Psi))  = \nabla \operatorname{Tr}(Q_i^{\prime} \Psi ) + \nabla \operatorname{Tr} (\Gamma^{i \top}_K(Q^{\prime}_i) )  \\ & \quad  \quad + \nabla \operatorname{Tr}(Q^{\prime \prime}_i \mathcal{T}_K^i(\Psi)) |_{Q^{\prime\prime} = \Gamma_K^{i\top}(Q^{\prime}_i)} \\
    & = -2B_i^{\top}[\sum_{t = 0}^{\infty} (\Gamma^{i,\top}_K)^t(Q^{\prime}_i)](A_i -B_i K)\mathcal{T}_K^i(\Psi) \\
    & = -2B_i^{\top} \mathcal{T}_K^{i, \top}(Q_i + K^{\top} R_i K) (A_i - B_i K) \Sigma^i_K, 
\end{aligned}
\end{equation*}
where we leverage the condition that spectrum $\rho(A_i - B_i K) < 1$, by which we have:
\begin{equation*}
    \operatorname{Tr}((\Gamma^{i,\top}_K)^t  Q^{\prime}_i) \leq \| Q_i^{\prime} \| \| A_i - B_i K\|^{2 t}  \underset{t \to \infty}{ \rightarrow }0, 
\end{equation*}
thus the series converge.
Combining with the fact that $P_K^i$ is actually the solution to the fixed point equation:
$   P^i_K = \mathcal{T}_K^i (Q_i + K^{\top} R_i K)$,
we get the desired result. 
\begin{equation*}
\begin{aligned}
   \nabla J_i(K) & = 2\left[\left(R_i+B_i^{\top} P_{K} B_i\right) K-B_i^{\top} P^i_{K} A_i\right] \Sigma^i_{K}. 
\end{aligned}
\end{equation*}

Now, we let the the Hessian $\nabla^2 J_i (K)[K]$ act on an arbitrary $X \in \R^{d\times k}$, decomposing the gradient $\nabla J_i (K) = f_1(K) f_2(K)$, we have:
\begin{align*}
    \nabla^2 J_i(K)  & = f^{\prime}_1(K) f_2(K) + f_1(K) f_2^{\prime}(K)  \\
  \nabla^2 J_i (K) [X]  & = f^{\prime}_1(K)[X] f_2(K) [X] + f_1(K)[X] f_2^{\prime}(K)[X].
\end{align*}
Hence, 
\begin{align*}
    f^{\prime}_1(K) f_2(K)[X, X] & = 2\left\langle\left(R_i X +B^{\top}_i X B_i X -B^{\top}_i P^{i, \prime}_K(K)[X] (A_i - B_i K)\right) \Sigma^i_K, X\right\rangle \\
    f_1(K) f^{\prime}_2(K)[X, X] & = 2\left\langle\left(R_i K-B^{\top}_i P^i_K (A_i - B_i K)\right)\Sigma^{i, \prime}_K(K)[X], X\right\rangle
\end{align*}
where $P^{i, \prime}_K [X]$ satisfies, let $A_K = A_i - B_i K$:
\begin{equation*}
     A_{K}^{\top} P^i_K(-B X)+(-B_i X)^{\top} P^i_K A_{K}+A_{K}^{\top}\left(P^{i,\prime}_K(K)[E]\right) A_{K}+X^{\top} R_i K+K^{\top} R_i X=P^{i,\prime}_K(K)[E]
\end{equation*}
and 
\begin{equation*}
    \Sigma^{i,\prime}_K(K)[X]=(-B_i X) \Sigma^i_K A_{K}^{\top}+A_{K} \Sigma^i_K(-B_i X)^{\top}+A_{K}\left(\Sigma^{i,\prime}_K(K)[X]\right) A_{K}^{\top}. 
\end{equation*}
Observing that the above expressions can be written as:
\begin{align*}
    P^{i,\prime}_K(K)[X]&=\sum_{j=0}^{\infty}\left(A_{K}^{\top}\right)^{j}\left(\left(K^{\top} R_i-A_{K}^{\top} P^i_K B_i\right) X+X^{\top}\left(R_i K-B_i^{\top} P^i_K A_{K}\right)\right)\left(A_{K}\right)^{j}, \\
    \Sigma^{i,\prime}_K(K)[X] &=\sum_{j=0}^{\infty}\left(A_{K}\right)^{j}\left(-B_i X \Sigma^i_K A_{K}^{\top}-A_{K} \Sigma^i_K X^{\top} B_i^{\top}\right)\left(A_{K}^{\top}\right)^{j} ,
\end{align*}
if $K$ is a stable policy. With the cyclic property of the matrix trace, we observe that:
\begin{equation*}
     \left.\left\langle B^{\top}_i\left(P^{i,\prime}_K(K)[X]\right) A_{K}\right) \Sigma^i_K, X\right\rangle=\left\langle\left(B^{\top} P^i_K A_{K}-R_i K\right)\left(\Sigma^{i,\prime}_K(K)[X]\right), X\right\rangle, 
\end{equation*}
and hence simplifying the expression as:
\begin{equation*}
    \nabla^2 J_i(K) = 2 (R_i X + B_i^{\top} P^i_K B X) \Sigma^i_K - 4 (B_i^{\top} P^{i,\prime}_K(K)[X] (A_i -B_iK)) \Sigma^i_K. 
\end{equation*}
Since $\nabla^2 J_i(K)$ is self adjoint, it is not hard to characterize the operator norm as
\begin{equation*}
     \|\nabla J_i (K) \|^2 = \sup_{\|X\|_F = 1 } \| \nabla^2 J_i(K) [X]\|^2_F = \sup_{\|X\|_F = 1 } \left( \nabla^2 J_i(K) [X, X] \right) ^2.
\end{equation*}
\qedhere
}

\section{Auxiliary Results}  \label{sec:appa}
This section presents several essential lemmas and norm inequalities that serve as fundamental tools in analyzing the stability and convergence properties of the learning framework, which have been also frequently revisited in the literature. 
These results essentially capture the local smoothness and boundedness properties of the costs and gradients for LQR tasks, we explicitly define the positive polynomials $h_G(K), h_c(K), h_H(K), h_{\Delta}(K), h_{{cost}}(K)$,  $h_{{grad}}(K)$, $h_{\mathcal{L}, G}(K)$, and $h_{\mathcal{L}, grad}(K)$ which are slightly adjusted version of those in \cite{toso2024meta,wang2023fedsysid}.

Throughout the paper, we use $\bar{\cdot}$ and $\underline{\cdot}$ to denote the supremum and infimum of some positive polynomials, e.g., $\bar{h}:= \sup_{K \in \mathcal{S}} {h}(K)$ and $\underline{h}:= \inf_{K \in \mathcal{S}} {h}(K)$ are the supremum and infimum of $h(K)$ over the set of stabilizing controllers $\mathcal{S}$, when we consider a set of $M$ matrices $\{A_i\}_{i=1}^M$, we denote $\|A\|_{\max} := \max_{i} \|A_i\|$, and  $\|A\|_{\min} := \min_{i} ||A_i||$. 

We may repeatedly employ Young's inequality and Jensen's inequality:
\begin{itemize}
    \item (Young's inequality)Given any two matrices $A, B \in \mathbb{R}^{n_x\times n_u}$, for any $\beta>0$, we have
\begin{align}\label{eq:youngs}
\|A+B\|_2^2 &\leq(1+\beta)\|A\|_2^2+\left(1+\frac{1}{\beta}\right)\|B\|_2^2 \leq (1+\beta)\|A\|_F^2+\left(1+\frac{1}{\beta}\right)\|B\|_F^2.
\end{align}
Moreover, given any two matrices $A, B$ of the same dimensions,  for any $\beta>0$, we have
\begin{align}\label{eq:youngs_inner_product}
\langle A, B\rangle & \leq \frac{\beta}{2}\lVert A\rVert_2^2 +\frac{1}{2\beta}\lVert B \rVert_2^2  \leq  \frac{\beta}{2}\lVert A\rVert_F^2 +\frac{1}{2\beta}\lVert B \rVert_F^2.
\end{align}
\item (Jensen's inequality) Given $M$ matrices $A^{(1)}, \ldots, A^{(M)}$ of identical dimensions, we have that 
\begin{align}\label{eq:sum_expand}
\left\|\sum_{i=1}^M A^{(i)}\right\|_2^2 \leq M \sum_{i=1}^M\left\|A^{(i)}\right\|_2^2, 
\left\|\sum_{i=1}^M A^{(i)}\right\|_F^2 \leq M \sum_{i=1}^M\left\|A^{(i)}\right\|_F^2.
\end{align}
\end{itemize}

\begin{lemma}[Uniform bounds \cite{toso2024meta}]\label{applem:uniformbounds}
Given a LQR task $\mathcal{T}_i$ and an stabilizing controller $K \in \mathcal{S}$, the Frobenius norm of gradient $\nabla J_i(K)$, Hessian $\nabla^2 J_i(K)$ and control gain $K$ can be bounded as follows:
$$
\|\nabla J_i(K)\|_F \leq h_G(K), \text{ } \|\nabla^2 J_i(K)\|_F \leq h_H(K), \text { and } \|K\|_F \leq h_c(K),
$$
with 
\begin{align*}
    h_G(K)&= \frac{J_{\max}(K) \sqrt{\frac{\max_{i}\left\|R_i+B^{\top}_i P^i_K B_i\right\|\left(J_{\max}(K)-J_{\min}\right)}{\mu}}}{\min_i\sigma_{\min}(Q_i)},\\
    h_H(K)&= \left(2\|R\|_{\max} + \frac{2\|B\|_{\max}{J}_{\max}(K)}{\mu} + \frac{4\sqrt{2}\tilde{\xi}_{\max}\|B\|_{\max}{J}_{\max}(K)}{\mu}\right)\frac{{J}_{\max}(K)k}{\|Q\|_{\min}}, \\
    h_c(K)&=\frac{\sqrt{\frac{\max_{i}\left\|R_i+B^{\top}_i P^{i}_K B_i\right\|\left(J_{\max}(K)-J_{\min}\right)}{\mu}}+\left\|B^{\top}_i P^i_K A_i\right\|_{\max}}{\sigma_{\min}(R)},
\end{align*}
with $\tilde{\xi}_{\max} := \frac{1}{\|Q\|_{\min}}\left( \frac{(1 + \|B\|^2_{\max}){J}_{\max}(K_0)}{\mu} + \|R\|_{\max} - 1\right)$.
 \end{lemma}

\begin{proof} See \cite{fazel2018global,wang2023fedsysid}. For $\|\nabla^2 J_i\|_F$, see in \cite[Lemma 7.9]{bu2019lqrlensordermethods}. 
\end{proof}

\begin{lemma}[Perturbation Analysis \cite{toso2024meta,musavi2023convergence}]\label{applem:perturbationanalysis}
    Let $K, K^{\prime} \in \mathcal{S} $ such that $\|\Delta\| := \|K^{\prime} -K\|  \leq h_{\Delta}(K)\ <\infty$, then, we have the following set of local smoothness properties:
    \begin{equation*}
        \begin{aligned}
             &\left|J_i\left(K^{\prime}\right)-J_i(K)\right| \leq h_{\text {cost}}(K) J_i(K) \|\Delta\|_F, \\
&\left\|\nabla J_i\left(K^{\prime}\right)-\nabla J_i(K)\right\|_F \leq h_{\text {grad}}(K)\|\Delta\|_F, \\
&\left\|\nabla^2 J_i\left(K^{\prime}\right)-\nabla^2 J_i(K)\right\|_F \leq h_{\text {hess }}(K)\|\Delta\|_F,
        \end{aligned}
    \end{equation*}
    for all tasks $i \in \mathcal{T}$,
    where the problem-dependent parameters $h_{\text {cost}}(K), h_{\text {grad}}(K), h_{\text {hess}}(K)$ are listed as follows:
    \begin{align*}
    h_{\Delta}(K) &= \frac{\max_i \sigma_{\min}(Q_i) \mu}{4 ||B||_{\max} J_{\max}(K)\left(\left\|A-BK\right\|_{\max}+1\right)}, \\
    h_{\text {cost}}(K) &= \frac{4 \operatorname{Tr}\left(\Sigma_0\right)J_{\max}(K)\|R\|_{\max}}{\mu\min_i \sigma_{\min}\left(Q_i\right)}\left(\|K\|+\frac{h_{\Delta}(K)}{2}+\|B\|_{\max}\|K\|^2 \left(\left\|A -BK\right\|_{\max}+1\right) \nu(K)\right), \\
    h_{\text {hess }}(K) &= \sup _{\|X\|_F=1}2(h_1(K)+ 2h_2(K))\|X\|^2_F,\\
    h_{\text {grad}}(K) &= 4\left(\frac{J_{\max}(K)}{\min_i\sigma_{\min}(Q)}\right)\Big[\|R\|_{\max}+\|B\|_{\max}\left(\|A\|_{\max}+\|B\|_{\max}\left(\|K\|+h_{\Delta}(K)\right)\right)\nonumber\\
 &\times\left(\frac{h_{\text {cost }}(K) J_{\max}(K)}{\operatorname{Tr}\left(\Sigma_0\right)}\right)+ \|B\|^2_{\max}\frac{J_
{\max}(K)}{\mu}\Big]\\
&+8\left(\frac{J_{\max}(K)}{\min_i\sigma_{\min}(Q)}\right)^2\left(\frac{\|B\|_{\max}\left(\left\|A-BK\right\|_{\max}+1\right)}{\mu}\right)h_0(K).
\end{align*}
with $\nu(K) = \frac{J_{\max}(K)}{\min_i \sigma_{\min}(Q_i)\mu}$, $h_0(K) = \sqrt{\frac{\max_{i}\left\|R_i+B^{(i)\top} P^i_K B_i\right\|\left(J_{\max}(K)-J_{\min}\right)}{\mu}},$ and 
    \begin{align*}
    h_1(K) &= h_3(K)\|B\|^2_{\max} \frac{J_{\max}(K)}{\min_i\sigma_{\min }(Q_i)} + \Tilde{\mu} h_4(K)\|B\|_{\max} \frac{J_{\max}(K)}{\mu} + h_4(K) \max_i \operatorname{Tr}(R_i),\\
    h_2(K) &= \|B\|_{\max} J_{\max}(K)\left(\frac{h_6(K) h_4(K) \max_i \operatorname{Tr}\left(A_i - B_iK\right)}{\mu} + \|B\|_{\max} h_6(K) \Tilde{\mu} \nu(K)\right. \\ &\left.\hspace{9.7cm}+\frac{\Tilde{\mu} h_7(K)}{\min_i\sigma_{\min }(Q_i)}\right),\\
    h_3(K) &= 6\left(\frac{J_{\max}(K)}{\min_i \sigma_{\min }(Q_i)}\right)^2\|K\|^2\|R\|_{\max}\|B\|_{\max}(\|A-B K\|_{\max}+1 )\\
    &+6\left(\frac{J_{\max}(K)}{\min_i \sigma_{\min }(Q_i)}\right)\|K\|\|R\|_{\max},\\
    h_4(K) &= 4\left(\frac{J_{\max}(K)}{\min_i\sigma_{\min }(Q_i)}\right)^2 \frac{\|B\|_{\max}(\|A-B K\|_{\max}+1)}{\mu},\\
    h_6(K) &= \sqrt{\frac{1}{\min_i \sigma_{\min }(Q_i)}\left(\|R\|_{\max}+\frac{1+\|B\|^2_{\max}}{\mu} J_{\max}(K)\right)-1},\\
    h_7(K) &= 4\left(\nu(K) h_8(K) +8\nu^2(K)\|B\|_{\max}\left(\left\|A - BK\right\|_{\max}+1\right) h_9(K)\right),\\
    h_8(K)&=\|R\|_{\max}+\|B\|^2_{\max} \frac{J_{\max}(K)}{\mu} +\left(\|B\|_{\max}\|A\|_{\max}+\|B\|^2_{\max}\|K\|_{\max}\right) h_3(K),\\
    h_9(K)&= 2\left(\|R\|_{\max}\|K\|+\|B\|_{\max}\|A-B K\|_{\max} \frac{J_{\max}(K)}{\mu}\right).
\end{align*}
where  $\Tilde{\mu} = 1+\frac{\mu}{h_{\Delta}(K)}$.
\end{lemma}

\begin{proof}
    See  \cite[Appendix F]{wang2023fedsysid} and \cite[Lemma 7]{musavi2023convergence}. 
\end{proof}

\begin{lemma}[Gradient Domination]\label{applem:graddom}
For any system $i$, let $K^*_i$ be the optimal policy, let $K^\star$ be the MAML-optimal policy. Suppose $K \in \mathcal{S}$. Then, it holds that 
\begin{equation*}
   \begin{aligned}
   J_i(K)-J_i\left(K^{*}_i\right)  & \geq \mu \cdot\frac{\operatorname{Tr} \left(E_{K}^{i,\top} E^i_{K}\right)}{\left\|R_i+B_i^{\top} P^i_{K} B_i\right\| }  \\ J_i(K)-J_i\left(K^{*}_i\right) & \leq \frac{1} { \sigma_{\min }(R_i)} \cdot\left\|\Sigma^i_{K^{*}}\right\| \cdot \operatorname{Tr}\left(E_{K}^{i,\top} E^i_{K}\right) \\
 & \leq \frac{\left\|\Sigma^i_{K^{*}}\right\| }{ \mu^2 \sigma_{\min}(R_i)} \|\nabla J_i (K)\|^2_F = : \frac{1}{\lambda_i } \|\nabla J_i (K)\|^2_F 
   \end{aligned}
\end{equation*}
\end{lemma}

 \begin{proof}
     See \cite[Lemma 11]{fazel2018global}.
\end{proof}

\begin{lemma}\label{applem:uniformboundsl}
Given a prior $p$ over LQR task set $\mathcal{T}$, adaptation rate $\eta$, and an MAML stabilizing controller $K \in \mathcal{S}$, the Frobenius norm of gradient $\nabla \mathcal{L}(K)$ and control gain $K$ can be bounded as follows:
\begin{equation}
\begin{aligned}
       \| \nabla \mathcal{L}(K)\|_F & \leq h_{G, \mathcal{L}} (K) , 
\end{aligned}
\end{equation}
where $h_{G, \mathcal{L}}:= ( k + \eta h_H (K)) (1 + \eta h_{grad} (K)) h_G (K)  $ is dependent on the problem parameters.
\end{lemma}

\begin{proof}
    When $K \in \mathcal{S}$, by expression of $\nabla \mathcal{L}$, we have:
    \begin{align*}
        \|\nabla \mathcal{L}\|_F & = \| (I - \eta \nabla^2 J_i (K)) \nabla J_i (K - \eta \nabla J_i(K))\|_F \\
        & \leq  \|I - \eta \nabla^2 J_i (K) \|_F \| \nabla J_i (K - \eta \nabla J_i(K))\|_F \\
        & \leq  \left( \|I\|_F  - \eta \| \nabla^2 J_i (K)\|_F  \right)  \|\nabla J_i (K - \eta \nabla J_i(K)) - \nabla J_i (K ) + \nabla J_i (K)\|_F  \\
        & \leq (k+ \eta h_H(K)) ( 1 + \eta h_{grad}(K)) h_{G}(K) , 
    \end{align*}
    where we applied Young's inequality, triangle inequality, the Lipschitz property of $\nabla J$ and uniform bounds.
\end{proof}

\begin{lemma}[Perturbation analysis of $\nabla \mathcal{L}(K)$]\label{applem:perturbationl}
 Let $K, K^{\prime} \in \mathcal{S} $ such that $\|\Delta\| := \|K^{\prime} -K\|  \leq h_{\Delta}(K)\ <\infty$, then, we have the following set of local smoothness properties,
 \begin{align*} 
  &  |\mathcal{L}(K^{\prime}) - \mathcal{L}(K)| \leq h_{\mathcal{L},cost} \|\Delta\|_F\\
   &   \| \nabla \mathcal{L} (K) - \nabla \mathcal{L} (K^{\prime})\|_F   \leq h_{\mathcal{L},grad} \| \Delta\|_F,
 \end{align*}
 where $h_{\mathcal{L}, cost}: = h_{cost} (1+ \eta h_{grad}(K)) $ and $h_{\mathcal{L}, grad} := \eta h_{hess}(K) ( 1 + \eta h_{grad}) h_G (K) + (k + \eta h_H(K^{\prime})) h_{hess}(K) (1 + \eta h_{hess}(K)) $ are problem dependent parameters. 
\end{lemma}

\begin{proof}
    Suppose $K, K^{\prime} \in \mathcal{S} $ such that $\|\Delta\| := \|K^{\prime} -K\|  \leq h_{\Delta}(K)\ <\infty$. 
    For $\mathcal{L}$, we have:
    \begin{align*}
      & \quad \  | \mathcal{L}(K^{\prime}) - \mathcal{L}(K)| = |\mathbb{E}_{i \sim p}  J_i (K^{\prime} - \eta \nabla J_i(K^{\prime})) -\mathbb{E}_{i \sim p}  J_i (K - \eta \nabla J_i(K))| \\
      & \leq \mathbb{E}_{i \sim p} h_{\mathcal{L},cost} (\|\Delta\|_F +   \eta  \| \nabla J_i(K^{\prime}) - \nabla J_i(K) \|_F ) \\
      & \leq h_{\mathcal{L},cost}  ( 1 + \eta h_{grad} (K)) \| \Delta\|_F .
    \end{align*}
    For $\nabla \mathcal{L}$, we have: 
    \begin{align*}
       & \quad  \| \nabla \mathcal{L} (K) - \nabla \mathcal{L} (K^{\prime})\|_F \\
       & =  \| \mathbb{E}_{i \sim p}  (I - \eta \nabla^2 J_i (K^{\prime}) )\nabla J_i (K^{\prime} - \eta \nabla J_i(K^{\prime}))  - \mathbb{E}_{i \sim p}  (I - \eta \nabla^2 J_i (K) )\nabla J_i (K - \eta \nabla J_i(K))\|_F 
       \\ & \leq \mathbb{E}_{i \sim p} \| (I - \eta \nabla^2 J_i (K^{\prime}) )\nabla J_i (K^{\prime} - \eta \nabla J_i(K^{\prime}))  -  (I - \eta \nabla^2 J_i (K^{\prime}) ) \nabla J_i (K - \eta \nabla J_i(K)) \\ 
       & \quad + (I - \eta \nabla^2 J_i (K^{\prime}) )\nabla J_i (K - \eta \nabla J_i(K)) - (I - \eta \nabla^2 J_i (K) )J_i (K - \eta \nabla J_i(K))\|_F \\
       & \leq   \mathbb{E}_{i \sim p} \bigg[ \| I - \eta \nabla^2 J_i (K^{\prime}) \|_F \|\nabla J_i (K - \eta \nabla J_i(K)) - \nabla J_i (K^{\prime} - \eta \nabla J_i(K^{\prime}))\|_F \\
       & \quad +  \| \eta \nabla^2 J_i (K) - \eta \nabla^2 J_i (K^{\prime}) \|_F \|\nabla J_i (K - \eta \nabla J_i(K))\|_F \bigg] \\
       & \leq (k + \eta h_H(K^{\prime })) h_{grad} (1 + \eta h_{grad}(K) ) \|\Delta \|_F + \eta h_{hess} (K) (1 + \eta h_{grad}(K)) h_G (K) \|\Delta \|_F , 
    \end{align*}
    where we repeatedly applied norm inequalities, local Lipschitz continuity and uniform bounds. 
    
\end{proof}


 \begin{lemma}[Matrix Bernstein Inequality \cite{gravell2020learning}] Let $\left\{Z_i\right\}_{i=1}^m$ be a set of $m$ independent random matrices of dimension $d_1 \times d_2$ with $\mathbb{E}\left[Z_i\right]=Z$, $\left\|Z_i-Z\right\| \leq B_r$ almost surely, and maximum variance 

$$\max \left(\left\|\mathbb{E}\left(Z_i Z_i^{\top}\right)-Z Z^{\top}\right\|,\left\|\mathbb{E}\left(Z_i^{\top} Z_i\right)-Z^{\top} Z\right\|\right) \leq \sigma_r^2,$$ 
and sample average $\widehat{Z}:=\frac{1}{m} \sum_{i=1}^m Z_i$. Let a small tolerance $\epsilon \geq 0$ and small probability $0 \leq \delta \leq 1$ be given. If
$$
m \geq \frac{2 \min \left(d_1, d_2\right)}{\epsilon^2}\left(\sigma_r^2+\frac{B_r \epsilon}{3 \sqrt{\min \left(d_1, d_2\right)}}\right) \log \left[\frac{d_1+d_2}{\delta}\right]
$$
$\text { then } \mathbb{P}\left[\|\widehat{Z}-Z\|_F \leq \epsilon\right] \geq 1-\delta \text {.}$   
\label{lemma:Bernstein}
\end{lemma}

\begin{lemma}[Finite-Horizon Approximation]\label{lem:finitehorizon}
     For any $K$ such that $J_i(K)$ is well-defined for any $i \in [I]$, let the covariance matrix be $\Sigma_K^{i,(\ell)} := \mathbb{E}[ \frac{1}{\ell}\sum_{i=1}^\ell x_i x_i^{\top}]$ and $J_i^{(\ell)} (K) =  \mathbb{E}[ \frac{1}{\ell}\sum_{i=0}^\ell x_i^{\top} (Q_i + K^{\top} R_i K) x_i] $. If
     $$
        \ell  \geq \frac{d \cdot J^2_{\max} (K)}{ \epsilon \mu \sigma^2_{\min}(Q)} ,
     $$
     then $\|\Sigma^{i,(\ell)}_K - \Sigma^i_K\| \leq \epsilon $. Also, if 
     $$
        \ell \geq \frac{d \cdot J_{\max}^{2}(K)\left(\|Q\|_{\max} +\|R\|_{\max}\|K\|_{\max}^{2}\right)}{\epsilon \mu \sigma_{\min }^{2}(Q)},
     $$
     then $|J_i(K)  -   J^{(\ell)}_i (K)| \leq  \epsilon$.
\end{lemma}

\section{Controlling Gradient Estimation Error}\label{sec:appb}
In the following, we provide detailed proof of \Cref{lem:gradestimate} and \Cref{lem:metagradestimate}, which give the explicit sample requirements for the gradient/meta-gradient estimation to be close to the ground truth. Before proving, we first restate the results. 

\begin{lemma*}[Gradient estimation]
    For sufficiently small numbers $\epsilon, \delta \in (0,1)$, given a control policy $K$, let $\ell$, radius $r$, number of trajectories $M$ satisfying the following dependence, 
    \begin{align*}
        \ell & \geq h^1_{\ell} (\frac{1}{\epsilon}, \delta) : = \max\{ h_{\ell, grad}(\frac{1}{\epsilon}), h_{\ell, var}(\frac{1}{\epsilon}, \delta)\} \\
        r & \leq  h^1_r (\frac{1}{\epsilon}) := \min \{ 1/\bar{h}_{cost}, \underline{h}_{\Delta}, \frac{\epsilon}{4\bar{h}_{grad}}\} \\
        M & \geq h^1_M (\frac{1}{\epsilon}, \delta) := h_{sample}(\frac{4}{\epsilon}, \delta)
    \end{align*}
    Then, with probability at least $1 - 2\delta$, the gradient estimation error is bounded by
    \begin{equation}
        \| {\nabla} J_i (K)- \tilde{\nabla} J_i (K) \|_F \leq \epsilon, 
    \end{equation}
     for any task $i \in [I]$.
\end{lemma*}

\begin{proof}[Proof of \Cref{lem:gradestimate}]
 The goal of this lemma is to show that conditioned on a perturbed policy, in algo 2. $\widehat{K}^0_j = K^0 + U_j$ for some random sample index $j$, the gradient estimation and cost estimation have low approximation error with high probability. Now, we notice that this policy is perturbed but not adapted, (the meta-gradient estimation error is to characterize the gradient of the adapted policy). and define:
 \begin{equation*}
 \begin{aligned}
       \nabla_r J_i ({K})& = \frac{dk}{r^2} \E J_i ( {K} + U_m) U_m , \\
       \widehat{\nabla} & = \frac{1}{M}\sum_{m =1}^M \frac{dk}{r^2} J_i ( {K} + U_m)U_m ,  \\
    \tilde{\nabla} & = \frac{1}{M}\sum_{m=1}^M \frac{dk}{r^2 \ell} \sum_{l=1}^\ell (x_l)^{\top} (Q_i + R ({K} + U_m)) x_l .
 \end{aligned}
 \end{equation*}
 Then, for any stable policy $K$, the difference can be broken into three parts:
 \begin{equation*}
     \nabla J_i (K) - \tilde{\nabla} = \underbrace{ \bigg(\nabla J_i(K) - \nabla_r J_i (K) \bigg)}_{(i)} + 
   \underbrace{\left( \nabla_r J_i(K) - \widehat{\nabla} \right)}_{(ii)} +  \underbrace{\left( \widehat{\nabla} - \tilde{\nabla} \right)}_{(iii)}.
 \end{equation*}
 
 For $(i)$, we apply \Cref{applem:perturbationanalysis}, choosing the $r,\epsilon$ such that $\frac{\epsilon}{4} \geq \bar{h}_{grad} r \geq \bar{h}_{grad} \|U\|_F$, and $r \leq 1/\bar{h}_{cost}$, and $r \leq \underline{h}_{\Delta}$, then, for every $U$ on the sphere such that $\|U\|_F \leq r$. 
 We have $\|\nabla J_i(K+U) - \nabla J_i (K)\| \leq \frac{\epsilon}{4}$ for all tasks $i \in [I]$. Therefore, by Jensen inequality,  
 \begin{align*}
    \| \nabla_r J_i (K) - \nabla J_i (K)\|_F & \leq \mathbb{E}_{U \sim \mathbb{B}_r }\| \nabla J_i(K+U) - \nabla J_i(K)\|_F \leq \frac{\epsilon}{4} .
 \end{align*}

 For $(ii)$, we have $\mathbb{E}_{U \sim \mathbb{S}_r }[ \widehat{\nabla}] = \nabla_r J_i(K)$, each individual sample $Z_i := \frac{dk}{r^2} J_i (K+U_m) U_m $ is bounded. Let $\bar{J}_{\max}:= \sup_{K \in {\mathcal{S}_{\text{ML}}}} \max_{i} J_i(K)$, 
 \begin{align*}
     \| Z_i \|_F & \leq \frac{dk}{r^2} | J_i (K+ U_m)  - J_i(K) + J_i(K) |  \|U_m \|_F \\
     & \leq \frac{dk}{r^2} ( h_{cost} \bar{J}_{\max}\|U_m\|_F  + \bar{J}_{\max}  )r \\
     & = \frac{dk}{r}( 1 + r h_{cost} ) \bar{J}_{\max}
  \end{align*}
For $Z := \nabla_r J_i(K)$, 
\begin{align*}
    \|Z\|_F & \leq \mathbb{E}_{ U \sim \mathbb{B}_r } \| \nabla J(K+U) -  \nabla J(K)  + J(K)\| \\
    & \leq \mathbb{E}_{ U \sim \mathbb{B}_r } \| \nabla J(K+U) -  \nabla J(K) \| + \| \nabla J(K)\| \\
    & \leq   h_{grad}\|U\|_F + h_G (K)
    \\ & \leq  \bar{h}_{grad}r + \bar{h}_G .
\end{align*}
Hence, we can use triangle inequality and write, almost surely:
\begin{equation*}
    \|Z_i - Z\|_F \leq \|Z_i\|_F + \|Z\|_F \leq B_r := \frac{dk}{r} (1 + r \bar{h}_{cost})\bar{J}_{\max} + \bar{h}_{grad} r + \bar{h}_G, 
\end{equation*}

For the variance bound, we have
\begin{align*}
    \|\mathbb{E}(Z_iZ_i^\top)  - ZZ^\top\|_F &\leq \|\mathbb{E}(Z_iZ_i^\top)\|_F + \|ZZ^\top\|_F\\
    & \leq \max_{Z_i} (\|Z_i\|_F)^2 + \|Z\|^2_F\\
    &\leq \sigma^2_r := \left(\frac{dk}{r} (1+r\bar{h}_{{cost}}) \bar{J}_{\max}\right)^2 + \left(r \bar{h}_{grad} + \bar{h}_G \right)^2.
\end{align*}
  
 Applying matrix Bernstein inequality \Cref{lemma:Bernstein}, when 
 $$
 M \geq h_{sample}(\frac{4}{\epsilon}, \delta) := \frac{32 \min \left(d, k\right)}{\epsilon^2}\left(\sigma_r^2+\frac{B_r \epsilon}{12 \sqrt{\min \left(d, k\right)}}\right) \log \left[\frac{d+k}{\delta}\right],
 $$
 with probability at least $1-\delta$, 
 $$\| \nabla_r J_i(K) - \widehat{\nabla} \|_F \leq \epsilon/4.
 $$

 For $(iii)$,  by \Cref{lem:finitehorizon}, choosing the horizon length $\ell \geq h_{\ell, grad} :=  \frac{16d^2 k^2 \bar{J}_{\max}^2(\|Q\|_{\max} + \|R\|_{\max} \|K\|^2)}{\epsilon r \mu \sigma^2_{\min}(Q)}$, one has for any $K \in \mathcal{S}_{ML}$, 
 \begin{equation*}
     \| \frac{1}{M} \frac{dk}{r^2}\sum_{m = 1}^M J^{(\ell)}_i (K+U_m) U_m -  \frac{1}{M} \frac{dk}{r^2}\sum_{m = 1}^M J_i (K+U_m) U_m \|_F \leq \frac{\epsilon}{4} .
 \end{equation*}

 To finish the proof, one needs to show that with high probability, $J^{(\ell)}_i $ is close to $ \tilde{J}^{(\ell)}_i (K)= \frac{1}{\ell}\sum_{l=1}^\ell (x_l)^{\top} (Q_i + K^{\top}R_i K ) x_l = \operatorname{Tr}(\tilde{\Sigma}^i_K (Q_i + K^{\top} R_i K))$, therefore, one can show that the sample covariance $\tilde{\Sigma}^i_{K+U_m}$ concentrates, i.e., there exists a polynomial $h_{\ell,var}(\frac{4}{\epsilon}, \delta)$, (see \cite{fazel2018global} Lemma 32,) such that when $\ell \geq h_{\ell,var}(\frac{4}{\epsilon}, \delta)$, $\|\tilde{\Sigma}^i_{K+U_m} -\Sigma^{i, (\ell)}_{K+U_m}\| \leq \epsilon/(4\sigma_{\min}(Q_i) )$, thus $J^{(\ell)}_i - \tilde{J}^{(\ell)}_i (K)$ can be bounded, 
 \begin{align*}
     \| \frac{1}{M} \frac{dk}{r^2}  \sum_{m=1}^M (J^{(\ell)}_i (K+U_m)U_m -  \frac{1}{M} \frac{dk}{r^2}  \sum_{m=1}^M\tilde{J}^{(\ell)}_i (K + U_m)U_m ) \|_F  \leq \frac{\epsilon}{4}. 
 \end{align*}
Adding all four terms together finishes the proof.

    \qedhere
\end{proof}

\begin{lemma*}
    For sufficiently small numbers $\epsilon, \delta \in (0,1)$, given a control policy $K$, let $\ell$, radius $r$, number of trajectory $M$ satisfies that 
    \begin{align*}
      | \mathcal{T}_n |& \geq  h_{sample, task} (\frac{2}{\epsilon}, \frac{\delta}{2}) , \\
        \ell & \geq  \max\{h^1_{\ell} (\frac{1}{\epsilon^{\prime}}, \delta^{\prime}) ,  h^2_{\ell, grad} (\frac{12}{\epsilon}), h^2_{\ell, var}(\frac{12}{\epsilon}, \delta^{\prime})\} , \\
        r & \leq \min \{ h^2_r(\frac{6}{\epsilon}),  h^1_r (\frac{1}{\epsilon}) \}  , \\
        M & \geq \max\{ h^2_M (\frac{1}{\epsilon}, \delta), h^1_M (\frac{1}{\epsilon^{''}}, \frac{\delta}{4})   \} , 
    \end{align*}
where $ h^2_M (\frac{1}{\epsilon}, \delta):= h_{sample}(\frac{1}{\epsilon^{''}}, \frac{\delta^{\prime}}{4})$, $\delta^{\prime} = \delta / h_{sample, task} (\frac{2}{\epsilon}, \frac{\delta}{2} ) $, $\epsilon^{\prime} = \frac{\epsilon}{6 \frac{dk}{r} h_{cost} \bar{J}_{max}}$, $\epsilon^{''} = \frac{\epsilon}{6}$.
 Then, for each iteration the meta-gradient estimation is $\epsilon$-accurate, i.e., 
 \begin{align*}
     \|\tilde{\nabla} \mathcal{L} (K) - \nabla \mathcal{L} (K) \|_F \leq \epsilon
 \end{align*} 
 with probability at least $1 - \delta$. 
\end{lemma*}

\begin{proof}[proof of \Cref{lem:metagradestimate}]
Again, the objective of this lemma is to show how accurate the meta gradient estimation is when the learning parameters are properly chosen. Essentially, we want to control $\|\tilde{\nabla} \mathcal{L} (K) - \nabla \mathcal{L} (K)\|$, where $  \mathcal{L} (K) := \mathbb{E}_{i \sim p} [\mathcal{L}_i (K)] $, we define the following quantities:
    \begin{align*} 
    \tilde{\nabla} \mathcal{L} (K) & = \frac{1}{|\mathcal{T}_n|} \sum_{i \in \mathcal{T}_n }\tilde{ \nabla} \mathcal{L}_i (K) \\
      \nabla \mathcal{L}_i (K) &= \nabla  J_i (K - \eta \nabla J_i (K)) \\
        \nabla_r \mathcal{L}_i (K) & = \frac{dk}{r^2} \mathbb{E}_{U \sim \mathbb{S}_r} [J_i (K + U - \eta \nabla J_i (K+U)) U] \\
        \widehat{\nabla}_r \mathcal{L}_i (K) & = \frac{dk}{r^2} \mathbb{E}_{U \sim \mathbb{S}_r} [J_i (K + U - \eta \tilde{\nabla} J_i (K+U)) U]  \\
        \tilde{\nabla} \mathcal{L}_i (K) & = \frac{dk}{r^2} \sum_{m=1}^M \tilde{J}^{(\ell)}_i(K+U_m - \eta \tilde{\nabla} J_i(K+U_m) ) U_m .
    \end{align*}
Then, similar to the proof of Lemma \ref{lem:gradestimate} we are able to break the gradient estimation error into two parts:
\begin{align*}
     \|\tilde{\nabla} \mathcal{L} (K) - \nabla \mathcal{L} (K) \| & \leq \| \mathbb{E}_{i \sim p} [\nabla \mathcal{L}_i (K)] - \frac{1}{|\mathcal{T}_n|} \sum_{i \in \mathcal{T}_n }\nabla \mathcal{L}_i (K) \|  \\
     & \quad + \frac{1}{|\mathcal{T}_n|}\sum_{i \in \mathcal{T}_n }\|\nabla \mathcal{L}_i(K) - \tilde{\nabla} \mathcal{L}_i(K) \| . 
\end{align*}

The first term is the difference between the sample mean of meta-gradients across different tasks, we apply matrix Bernstein \Cref{lemma:Bernstein} to show that when the task batch size $|\mathcal{T}_n|$ is large enough, with probability $\frac{\delta}{2}$,
\begin{equation*}
    \| \frac{1}{|\mathcal{T}_n|} \sum_{i \in \mathcal{T}_n }{ \nabla} \mathcal{L}_i (K)  - \mathbb{E}_{i \sim p} \nabla  \mathcal{L}_i (K)\|_F \leq  \frac{\epsilon}{2}.
\end{equation*} 
We begin with the expression of the meta-gradient:
\begin{align*}
   \nabla \mathcal{L}_i (K) =  (I - \eta\nabla^2 J_i (K) ) \nabla J_i (K - \eta \nabla J_i (K)), 
\end{align*} 
and let an individual sample be $X_i = \nabla \mathcal{L}_i (K)$, and $X = \mathbb{E}_{i \sim p} \nabla \mathcal{L}_i (K)$, then, it is not hard to establish the following using \Cref{applem:uniformbounds}:
 \begin{align*}
      \|X_i\|_F \leq (1 + \eta \bar{h}_H) \bar{h}_G  \quad \quad \|X\|_F \leq (1 + \eta \bar{h}_H) \bar{h}_G .
 \end{align*}
Thus, 
 \begin{align*}
     \|X - X_i\|_F \leq B_{\mathcal{T}} & :=  2(1+\eta \bar{h}_H) \bar{h}_G  \quad almost \ surely, \\
     \|\mathbb{E}(X_iX_i^\top)  - XX^\top\|_F &\leq \|\mathbb{E}(X_iX_i^\top)\|_F + \|XX^\top\|_F\\
    & \leq \max_{X_i} (\|X_i\|_F)^2 + \|X\|^2_F \\
    & \leq \sigma_{\mathcal{T}_n}^2 := 2(1+\eta \bar{h}_H)^2 \bar{h}^2_G .
  \end{align*}
  Therefore, the final requirement is for the task batch size to be sufficient: $$|\mathcal{T}_n| \geq  h_{sample, task} (\frac{2}{\epsilon}, \frac{\delta}{2}) := \frac{8 \min \left(d, k\right)}{\epsilon^2}\left(\sigma_{\mathcal{T}}^2+\frac{B_{\mathcal{T}} \epsilon}{6 \sqrt{\min \left(d, k\right)}}\right) \log \left[\frac{2(d+k)}{\delta}\right].$$

For the second term $\frac{1}{|\mathcal{T}_n|}\sum_{i \in \mathcal{T}_n } \| \nabla \mathcal{L}_i(K) - \tilde{\nabla} \mathcal{L}_i(K)\|$,  we bound each task-specific difference individually, which can be bounded as the following using triangle inequality:
    \begin{align*}
     \| \nabla  - \tilde{\nabla} \| \leq \underbrace{\| \nabla - \nabla_r\|}_{(i)} + \underbrace{\| \nabla_r - \widehat{\nabla}_r\|}_{(ii)} + \underbrace{\| \widehat{\nabla}_r - \tilde{\nabla} \|}_{(iii)}.
    \end{align*}
To quantify $(i)$ is to quantify the difference between $\nabla J_i (K - \eta \nabla J_i (K))$ and $ \nabla_r \mathcal{L}_i \equiv \nabla J_i(K + U- \eta \nabla J_i (K+U))$, when $U$ is uniformly sampled from the $r$-sphere.
Applying \Cref{applem:perturbationanalysis} and \Cref{applem:uniformbounds}, we have
\begin{align*}
 &  \quad \quad \|  \nabla \mathcal{L}_i (K+U) -  \nabla \mathcal{L}_i(K) \|_F \\
 & = \|(I - \eta \nabla^2  J_i (K+ U)) \nabla J_i (K + U - \eta \nabla J_i (K+U))  \\ 
 & \quad\quad\quad\quad - (I - \eta \nabla^2 J_i (K)) \nabla J_i (K - \eta \nabla J_i (K)) \|_F \\
& = \| \left(  (I - \eta \nabla^2  J_i (K+ U) ) - ( I - \eta \nabla^2 J_i (K) ) \right) \nabla J_i (K + U - \eta \nabla J_i (K+U))\|_F \\
& \quad +\| ( I - \eta \nabla^2 J_i (K) ) \left( \nabla J_i (K - \eta \nabla J_i (K)) - \nabla J_i (K + U - \eta \nabla J_i (K+U) \right)\|_F \\
& \leq   \eta \bar{h}_{hess} r \bar{h}_G  +  (1 + \eta h_H) h_{grad} (1 + \eta h_{grad})r \\
& = (\eta \bar{h}_{hess} \bar{h}_G +(1 + \eta h_H) (1 + \eta h_{grad}) h_{grad} ) r
\end{align*}
Let $r \leq h^2_r (\frac{6}{\epsilon}) := \frac{1}{6(\eta \bar{h}_{hess} \bar{h}_G +(1 + \eta h_H + \eta h_{grad}  + \eta^2 h_{H}h_{grad})  h_{grad} )}$, we arrive at $(i) \leq \frac{\epsilon}{6}$.

For $(ii)$, as we have established in \Cref{lem:gradestimate}, for each task $i$, as long as the parameters $\ell, r,$ and $M$ are bounded by certain polynomials, with probability $1-\delta$, $\| \nabla J_i  - \tilde{\nabla} J_i \|_F \leq \epsilon^{\prime}$, which enables us to apply the perturbation analysis \Cref{applem:perturbationanalysis} again, 
\begin{align*}
     \| \nabla_r \mathcal{L}_i (K) - \widehat{\nabla}_r  \mathcal{L}_i (K)\|_F \leq \frac{dk}{r} h_{cost} \bar{J}_{max}\epsilon^{\prime}.
\end{align*}
Let $\frac{\epsilon}{6}  = \frac{dk}{r} h_{cost} \bar{J}_{max}\epsilon^{\prime}$, we obtain that once $r \leq h^1_r (1/\epsilon^{\prime} )$, $\ell \geq h^1_{\ell}(1/\epsilon^{\prime}, \frac{\delta^{\prime}}{4})$, and $M \geq h^1_M(1/\epsilon^{\prime}, \frac{\delta^{\prime}}{4})$, it holds that $(ii) \leq \frac{\epsilon}{6}$ with probability $1 - \frac{\delta}{2}$.

For $(iii)$, the analysis is identical to the analysis for $(ii)+(iii)$ plus the finite horizon approximation error in the proof of \Cref{lem:gradestimate}, except that the cost function $J_i$ is evaluated at $K - \eta \tilde{\nabla} J_i (K)$, but the uniform bounds \Cref{applem:uniformbounds} still apply here.
We hereby define each individual sample $Z_i := \frac{dk}{r^2} J_i (K+U_m - \eta \tilde{\nabla} J_i (K+ U_m)) U_m  $ and the mean $Z := \mathbb{E}_{ U \sim \mathbb{B}_r } \nabla J(K+U- \eta \tilde{\nabla} J_i (K+ U))$. For $Z_i$, we have:
\begin{align*}
      \| Z_i \|_F & \leq \frac{dk}{r^2} | J_i (K+ U_m - \eta\tilde{\nabla} J_i (K+ U_m))  - J_i(K - \eta\tilde{\nabla} J_i (K)) 
      \\ & \quad + J_i(K - \eta\tilde{\nabla} J_i (K)) |  \|U_m \|_F \\
     & \leq \frac{dk}{r^2} ( h_{cost} \bar{J}_{\max} (1+\eta \frac{dk}{r}(\bar{h}_{G} + \epsilon^{''})\|U_m \|_F  + \bar{J}_{\max}  )r \\
     & = \frac{dk}{r}( 1 + r h_{cost} (1+\eta \frac{dk}{r} (\bar{h}_{G} + \epsilon^{''})) \bar{J}_{\max} ,
\end{align*}
where the second inequality requires the Lipshitz analysis of the composite function, where the inner function $ \tilde{K} = K - \eta\tilde{\nabla} J_i$ has a Lipshitz constant $1+\eta \frac{dk}{r}  (\bar{h}_{G} + \epsilon^{''})$, where $\epsilon^{''} = \frac{\epsilon}{6}$ is depending on the parameters for the inner loop. 
For $Z$, we have:
\begin{align*}
     \|Z\|_F & \leq \mathbb{E}_{ U \sim \mathbb{B}_r } \| \nabla J(K+U - \eta \tilde{\nabla} J_i (K+U)) -  \nabla J(K- \eta \tilde{\nabla} J_i (K))  \\
     & \quad + \nabla J(K - \eta \tilde{\nabla} J_i (K))\|_F \\
    & \leq \mathbb{E}_{ U \sim \mathbb{B}_r } \| \nabla J(K+U - \eta \tilde{\nabla} J_i (K+U) ) -  \nabla J(K - \eta \tilde{\nabla} J_i (K)) \|_F  \\
    & \quad + \| \nabla J(K- \eta \tilde{\nabla} J_i (K))\|_F \\
    & \leq   h_{grad}( 1 + \eta \frac{dk}{r} (\bar{h}_{H}+ \epsilon^{''}))\|U\|_F + h_G (K - \eta \tilde{\nabla} J_i(K)) 
    \\ & \leq  \bar{h}_{grad}( 1 + \eta \frac{dk}{r} (\bar{h}_{H}+ \epsilon^{''})) r + \bar{h}_G .
\end{align*}

Therefore the new $B_r$ and $\sigma_r$ can be bounded as:
\begin{align*}
     B_r & := \frac{dk}{r}( 1 + r h_{cost} (1+\eta \frac{dk}{r} (\bar{h}_{G} + \epsilon^{''})) \bar{J}_{\max} + \bar{h}_{grad}( 1 + \eta \frac{dk}{r} (\bar{h}_{H}+ \epsilon^{''})) r + \bar{h}_G
     \\
     \sigma_r & := \left( \frac{dk}{r}( 1 + r h_{cost} (1+\eta \frac{dk}{r} (\bar{h}_{G} + \epsilon^{''})) \bar{J}_{\max}\right)^2 + \left(\bar{h}_{grad}( 1 + \eta \frac{dk}{r} (\bar{h}_{H}+ \epsilon^{''})) r + \bar{h}_G\right)^2 .
\end{align*}
Applying matrix Bernstein inequality \Cref{lemma:Bernstein} again, when 
 $$
 M \geq h^2_{M} (\frac{1}{\epsilon}, \delta) := h_{sample}(\frac{1}{\epsilon^{''}}, \frac{\delta^{\prime}}{4}) := \frac{96 \min \left(d, k\right)}{\epsilon^2}\left(\sigma_r^2+\frac{B_r \epsilon}{18 \sqrt{\min \left(d, k\right)}}\right) \log \left[4\frac{d+k}{\delta^{\prime}}\right],
 $$
 with probability at least $1-\frac{\delta^{\prime}}{4}$, for any $K \in \mathcal{S}_{ML}$, 
 $$\| \nabla_r \mathcal{L}_i(K) - \frac{dk}{r^2} \sum_{m=1}^M {J}_i(K+U_m - \eta \tilde{\nabla} J_i(K+U_m) ) U_m \|_F \leq \epsilon/6.
 $$
Again by previous analysis, we choose here the horizon length $\ell \geq 
 h^2_{\ell, grad} (\frac{12}{\epsilon}) := \frac{32d^2 k^2 \bar{J}_{\max}^2(\|Q\|_{\max} + \|R\|_{\max} \|K\|^2)}{\epsilon r \mu \sigma^2_{\min}(Q)}$ and $\ell \geq h_{\ell, var} (\frac{12}{\epsilon}, \frac{\delta^{\prime}}{4})$, so that the following two hold with probability $1 - \frac{\delta^{\prime}}{4}$:
 \begin{align*}
      \| \frac{1}{M} \frac{dk}{r^2}\sum_{m = 1}^M J^{(\ell)}_i (K+U_m&   - \eta \tilde{\nabla} J_i(K+U_m)) U_m \\
      & \quad -  \frac{1}{M} \frac{dk}{r^2}\sum_{m = 1}^M J_i (K+U_m  - \eta \tilde{\nabla} J_i(K+U_m)) U_m \|_F \leq \frac{\epsilon}{12} \\
       \| \frac{1}{m} \frac{dk}{r^2}  \sum_{m=1}^M J^{(\ell)}_i (K+U_m &- \eta \tilde{\nabla} J_i(K+U_m))U_m \\
       & \quad -  \frac{1}{M} \frac{dk}{r^2}  \sum_{m=1}^M \tilde{J}^{(\ell)}_i (K + U_m- \eta \tilde{\nabla} J_i(K+U_m))U_m  \|_F  \leq \frac{\epsilon}{12} .
 \end{align*}
Hence, we arrive at, with high probability $1 - \delta^{\prime}$, 
\begin{equation*}
    \|\nabla \mathcal{L}_i (K) - \tilde{\nabla}\mathcal{L}_i (K) \|_F \leq  \frac{1}{2}\epsilon .
\end{equation*} 
The proof is finished by letting $\delta^{\prime} = \delta / h_{sample, task} (\frac{2}{\epsilon}, \frac{\delta}{2} ) $, and applying a a union bound argument.

\end{proof}

\section{Theoretical Guarantees}\label{sec:appc}

\begin{theorem} \label{appthm:stabilityattain}
Given an initial stabilizing controller $K_0 \in \mathcal{S}$ and scalar $\delta \in (0,1)$,  let $\varepsilon_i := \frac{\lambda_i \Delta^i_0}{6}$, the adaptation rate $\eta \leq \min\{ \sqrt{\frac{1}{4(\bar{h}_{grad}^2 k^2 + \bar{h}_{grad}^2 \bar{h}_{H}^2 + \bar{h}_H^2 )}} , \frac{1}{4\bar{h}_{\text{grad}}}\}$, 
and $\varepsilon :=  \frac{\bar{\lambda}_i  \bar{\Delta}^i_0 (1 -2\phi_1)\phi_2 }{2 (1 + 4\phi_2 - 2\phi_1)}$ where $\phi_1 := 2(k^2 + \eta^2\bar{h}_H^2 )\eta^2 \bar{h}^2_{grad} + 2 \eta^2 \bar{h}_H^2$ and $\phi_2 := (k^2 + \eta^2 \bar{h}_H^2) (2 + 2 \bar{h}^2_{grad} \eta^2$, and the learning rate $\alpha \leq \frac{\frac{1}{2} - \phi_1}{2 \phi_2 \bar{h}_{grad}}$. 
In addition, the task batch size $|\mathcal{T}_n|$, the smoothing radius $r$, roll-out length $\ell$, and the number of sample trajectories satisfy:
\begin{align*}
      | \mathcal{T}_n |& \geq  h_{sample, task} (\frac{2}{\varepsilon}, \frac{\delta}{2}) , \\
        \ell & \geq  \max\{h^1_{\ell} (\frac{1}{\varepsilon_i}, \frac{\delta}{2}), h^1_{\ell} (\frac{1}{\varepsilon^{\prime}}, \delta^{\prime}) ,  h^2_{\ell, grad} (\frac{12}{\varepsilon}), h^2_{\ell, var}(\frac{12}{\varepsilon}, \delta^{\prime})\} , \\
        r & \leq \min \{ h^1_r (\frac{1}{\varepsilon_i}),  h^1_r (\frac{1}{\varepsilon}), h^2_r(\frac{6}{\varepsilon}) \}  , \\
        M & \geq \max\{ h^1_M(\frac{1}{\varepsilon_i}, \frac{\delta}{2}), h^1_M (\frac{1}{\varepsilon^{''}},\frac{\delta}{4})  h^2_M (\frac{1}{\varepsilon}, \delta)  \} , 
    \end{align*}
where $ h^2_M (\frac{1}{\varepsilon}, \delta):= h_{sample}(\frac{1}{\varepsilon^{''}}, \frac{\delta^{\prime}}{4})$, $\delta^{\prime} = \delta / h_{sample, task} (\frac{2}{\varepsilon}, \frac{\delta}{2} ) $, $\varepsilon^{\prime} = \frac{\varepsilon}{6 \frac{dk}{r} h_{cost} \bar{J}_{max}}$, $\varepsilon^{''} = \frac{\varepsilon}{6}$. 
Then, with probability, $1-\delta$, $K^i_n, K_n \in \mathcal{S}$, for every iteration $\{0,1, \ldots, N\}$ of Algorithm \ref{metalqr}.
\end{theorem}

\proof[Proof]{
Our gradient of the meta-objective is estimated through a double-layered zeroth-order estimation, here we begin by showing that given a stabilizing initial controller $K_0 \in \mathcal{S}$, one may select $\eta, r$, $\ell$, and $M$ to ensure that it is also MAML stabilizing, i.e., $K^i_0 := K_0 - \eta \tilde{\nabla} J_i (K_0) \in \mathcal{S} \subseteq \mathcal{K}$ for all task $i$.
We start by using the local-smoothness smoothness property:
    \begin{align*}
   & \quad \   J_i({K}^i_0) - J_i (K_0) \\ 
   & \leq \langle \nabla J_i(K_0), {K}^i_0 - K_0 \rangle + \frac{\bar{h}_{{grad}}}{2}\|{K}^i_0 - K_0\|^2_F \\
     &= \langle \nabla J_i(K_0), - \eta \tilde{\nabla} J_i(K_0) \rangle + \frac{\bar{h}_{{grad}}\eta^2}{2}\|\tilde{\nabla} J_i(K_0)\|^2_F\\
     &\leq -\frac{\eta}{2}\|\nabla J_i(K_0)\|^2_F + \frac{\eta}{2}\|\tilde{\nabla} J_i(K_0) - \nabla J_i(K_0)\|^2_F + \frac{\bar{h}_{{grad}}\eta^2}{2}\|\tilde{\nabla} J_i(K_0)\|^2_F\\
     &\leq \left( \bar{h}_{{grad}}\eta^2 -\frac{\eta}{2} \right)\|\nabla J_i(K_0)\|^2_F + \left( \bar{h}_{{grad}}\eta^2 + \frac{\eta}{2} \right)\|\tilde{\nabla} J_i(K_0) - \nabla J_i(K_0)\|^2_F \\
     &\stackrel{(i)}{\leq} -\frac{\eta}{4} \|\nabla J_i(K_0)\|^2_F + \frac{3\eta}{4}\|\tilde{\nabla} J_i(K_0) - \nabla J_i(K_0)\|^2_F, 
\end{align*}
where inequality $(i)$ comes from the selection of $\eta \leq \frac{1}{4\bar{h}_{\text{grad}}}$. Note that this selection is for constructing a monotone recursion.
By \Cref{applem:graddom}, we can further bound the term $- \frac{\eta}{4}\| \nabla J_i (K_0)\|^2_F \leq -\frac{\eta \lambda_i}{4} \left( J_i (K_0) - J_i(K^*_i) \right)$, rearranging the terms we get:
\begin{align*}
    & \quad \   J_i({K}^i_0) - J_i (K^*_i) \\ & \leq (1-\frac{\eta \lambda_i }{4}) \left( J_i (K_0) - J_i(K^*_i) \right)  + \frac{3\eta}{4}\|\tilde{\nabla} J_i(K_0) - \nabla J_i(K_0)\|^2_F,  \\ 
 &  = (1 - \frac{\eta \lambda_i}{4}) \Delta^i_0  + \frac{3\eta}{4}\|\tilde{\nabla} J_i(K_0) - \nabla J_i(K_0)\|^2_F
\end{align*}

Now the business is to characterize the distance between the estimated gradient $\tilde{\nabla}J_i(K_0)$ and $\nabla J_i(K_0)$. According to \Cref{lem:gradestimate}, let $\varepsilon_i  = \frac{\lambda_i \Delta^i_0}{6}$, when $\ell \geq h^1_{\ell} (\frac{1}{\varepsilon_i}, \frac{\delta}{2})$, $r \leq h^1_r (\frac{1}{\varepsilon_i})$ and $M \geq  h^1_M(\frac{1}{\varepsilon_i}, \frac{\delta}{2})$, $\|\tilde{\nabla} J_i(K_0) - \nabla J_i(K_0)\|^2_F \leq \varepsilon_i $ with probability $1 - \delta$, which leads to:
\begin{align*}
    J_i({K}^i_0) - J_i (K^*_i) & \leq (1 - \frac{\eta \lambda_i}{8}) \Delta^i_0 .
\end{align*}
Therefore, $J_i({K}^i_0) \leq J_i(K_0)$, which means that $K_0 - \eta  \tilde{\nabla}J_i(K_0) \in \mathcal{S}$. 

Now, we proceed to show that $K_1 \in \mathcal{S}$ as well. By smoothness property, we have that the meta-gradient update yields, $\forall n$:

\begin{align*}
    & \quad \quad \mathbb{E}_{i \sim p } [ J_i (K_{n+1}) - J_i (K_n) ] \leq 
 \langle  \mathbb{E}_{i \sim p } \nabla J_i (K_n), K_{n+1} - K_{n} \rangle + \frac{\bar{h}_{{grad}}}{2}\|K_{n+1} - K_n\|^2_F \\
    &=\langle  \mathbb{E}_{i \sim p } \nabla J_i (K_n), -\alpha \tilde{\nabla}  \mathcal{L}(K_n) \rangle + \frac{\bar{h}_{{grad}}\alpha^2}{2}\| \tilde{\nabla}  \mathcal{L}(K_n)\|^2_F\\
    & \leq -\frac{\alpha }{2}\|\mathbb{E}_{i \sim p } \nabla J_i (K_n)\|^2_F + \frac{\alpha}{2}\|\mathbb{E}_{i \sim p } \nabla J_i (K_n) - \tilde{\nabla} \mathcal{L}(K_n)\|^2_F + \frac{\bar{h}_{{grad}}\alpha^2}{2}\|\tilde{\nabla}  \mathcal{L}(K_n)\|^2_F \\
    & \leq -\frac{\alpha}{2}\|\mathbb{E}_{i \sim p } \nabla J_i (K_n)\|^2_F + \alpha \|\mathbb{E}_{i \sim p } \nabla J_i (K_n)  - \nabla \mathcal{L}(K_n)\|^2_F  \\
    & \quad +  ( \alpha + \alpha^2 \bar{h}_{grad}) \|\tilde{\nabla}  \mathcal{L}(K_n)  - \nabla \mathcal{L}(K_n)\|^2_F  + \alpha^2 \bar{h}_{grad} \|{\nabla}  \mathcal{L}(K_n)\|_F^2 .
\end{align*}
The perturbation analysis for the difference term $\|\mathbb{E}_{i \sim p } \nabla J_i (K_n)  - \nabla \mathcal{L}(K_n)\|^2_F$ and the uniform bounds on the gradients and Hessians show that, 
\begin{align*}
     \|\mathbb{E}_{i \sim p } \nabla J_i (K_n)  - \nabla \mathcal{L}(K_n)\|^2_F & \leq  (2(k^2 + \eta^2\bar{h}_H^2 )\eta^2 \bar{h}^2_{grad} + 2 \eta^2 \bar{h}_H^2) \|\nabla J_i (K_n )\|_F^2\\
     & := \phi_1 \|\mathbb{E}_{i \sim p} \nabla J_i (K_n)\|_F^2, \\
     \|{\nabla}  \mathcal{L}(K_n)\|_F^2 & \leq   (k^2 + \eta^2 \bar{h}_H^2) (2 + 2 \bar{h}^2_{grad} \eta^2 ) \|\mathbb{E}_{i \sim p} \nabla J_i (K_n) \|^2_F \\ & := \phi_2  \|\mathbb{E}_{i \sim p} \nabla J_i (K_n) \|_F^2 .  
 \end{align*}
Equipped with upper bounds above, we arrive at:
\begin{align*}
     & \quad \quad \mathbb{E}_{i \sim p }  [ J_i (K_{n+1}) - J_i (K_n) ] 
   \\&  \leq \alpha ( \phi_1 + \phi_2 \alpha \bar{h}_{grad}- \frac{1}{2}) \|\mathbb{E}_{i \sim p} \nabla J_i (K_n) \|_F^2   +  ( \alpha + \alpha^2 \bar{h}_{grad}) \|\tilde{\nabla}  \mathcal{L}(K_n)  - \nabla \mathcal{L}(K_n)\|^2_F \\
     & \stackrel{(ii)}{\leq} -\frac{\alpha}{2}(\frac{1}{2} -\phi_1) \| \mathbb{E}_{i \sim p} \nabla J_i (K_n) \|^2_F + \frac{\alpha (1 + 4\phi_2 - 2\phi_1)}{4\phi_2}\|\tilde{\nabla} \mathcal{L}(K_n) - \nabla \mathcal{L}(K_n)\|^2_F ,
\end{align*}
 where we select $\eta \leq \sqrt{\frac{1}{4(\bar{h}_{grad}^2 k^2 + \bar{h}_{grad}^2 \bar{h}_{H}^2 + \bar{h}_H^2 )}} $ to ensure $\phi_1 \leq \frac{1}{2}$, and $\alpha \leq \frac{\frac{1}{2} - \phi_1}{2 \phi_2 \bar{h}_{grad}}$ to arrive at inequality $(ii)$. By gradient domination property, 
 \begin{align*}
    & \quad \quad \mathbb{E}_{i \sim p }  [ J_i (K_1) - J_i (K^*_i) ] \\ & 
    \leq (1 - \frac{\bar{\lambda}_i (\alpha - 2 \alpha \phi_1)}{4}) \mathbb{E}_{i \sim p } [ J_i (K_{0}) - J_i (K^*_i)] +  \frac{\alpha (1 + 4\phi_2 - 2\phi_1)}{4\phi_2}\|\tilde{\nabla} \mathcal{L}(K_n) - \nabla \mathcal{L}(K_n)\|^2_F  \\
&  \leq (1 - \frac{\bar{\lambda}_i \alpha ( 1 - 2  \phi_1)}{4}) \bar{\Delta}^i_0  +  \frac{\alpha (1 + 4\phi_2 - 2\phi_1)}{4\phi_2}\|\tilde{\nabla} \mathcal{L}(K_n) - \nabla \mathcal{L}(K_n)\|^2_F 
 \end{align*}

 Now, we proceed to control the meta-gradient estimation error, according to \Cref{lem:metagradestimate}, let $\varepsilon :=  \frac{\bar{\lambda}_i  \bar{\Delta}^i_0 (1 -2\phi_1)\phi_2 }{2 (1 + 4\phi_2 - 2\phi_1)}$ when \begin{align*}
      | \mathcal{T}_n |& \geq  h_{sample, task} (\frac{2}{\varepsilon}, \frac{\delta}{2}) , \\
        \ell & \geq  \max\{h^1_{\ell} (\frac{1}{\varepsilon^{\prime}}, \delta^{\prime}) ,  h^2_{\ell, grad} (\frac{12}{\varepsilon}), h^2_{\ell, var}(\frac{12}{\varepsilon}, \delta^{\prime})\} , \\
        r & \leq \min \{ h^2_r(\frac{6}{\varepsilon}),  h^1_r (\frac{1}{\varepsilon}) \}  , \\
        M & \geq \max\{ h^2_M (\frac{1}{\varepsilon}, \delta), h^1_M (\frac{1}{\varepsilon^{''}}, \frac{\delta}{4})   \} , 
    \end{align*}
where $ h^2_M (\frac{1}{\varepsilon}, \delta):= h_{sample}(\frac{1}{\varepsilon^{''}}, \frac{\delta^{\prime}}{4})$, $\delta^{\prime} = \delta / h_{sample, task} (\frac{2}{\varepsilon}, \frac{\delta}{2} ) $, $\varepsilon^{\prime} = \frac{\varepsilon}{6 \frac{dk}{r} h_{cost} \bar{J}_{max}}$, $\varepsilon^{''} = \frac{\varepsilon}{6}$.
 Then, for each iteration the meta-gradient estimation is $\epsilon$-accurate, i.e., 
 \begin{align*}
     \|\tilde{\nabla} \mathcal{L} (K) - \nabla \mathcal{L} (K) \|_F \leq \varepsilon, 
 \end{align*} 
which leads to that
 \begin{align*}
      \mathbb{E}_{i \sim p }  [ J_i (K_1) - J_i (K^*_i) ]  
       \leq (1 - \frac{\bar{\lambda}_i \alpha ( 1 - 2  \phi_1)}{8}) \bar{\Delta}^i_0,
 \end{align*}
 with probability at least $1 - \delta$. This implies that $K_1 \in \mathcal{S}.$ 
 
 The stability is completed by applying induction steps for all iterations $n \in \{0,1,\ldots,N\}$, with the same analysis applies to every iteration.

\qedhere 
}

\begin{corollary}\label{appcor:convergence} (Convergence) 
Given an initial stabilizing controller $K_0 \in \mathcal{S}$ and scalar $\delta \in (0,1)$, let the parameters for Algorithm \ref{metalqr} satisfy the conditions in \Cref{thm:stabilityattain}. If, in addition, \begin{align*}
      | \mathcal{T}_n |& \geq  h_{sample, task} (\frac{2}{\bar{\varepsilon}}, \frac{\delta}{2}) , \\
        \ell & \geq  \max\{h^1_{\ell} (\frac{1}{\bar{\varepsilon}^{\prime}}, \delta^{\prime}) ,  h^2_{\ell, grad} (\frac{12}{\bar{\varepsilon}}), h^2_{\ell, var}(\frac{12}{\bar{\varepsilon}}, \delta^{\prime})\} , \\
        r & \leq \min \{ h^2_r(\frac{6}{\bar{\varepsilon}}),  h^1_r (\frac{1}{\bar{\varepsilon}}) \}  , \\
        M & \geq \max\{ h^2_M (\frac{1}{\bar{\varepsilon}}, \delta), h^1_M (\frac{1}{\bar{\varepsilon}^{''}}, \frac{\delta}{4})   \} , 
    \end{align*}
 where $\bar{\varepsilon} := \frac{\bar{\lambda}_i (1 - \eta^2 \bar{h}_H^2)\psi_0}{6}$, $\psi_0 :=  \mathcal{L} (K_0 ) - \mathcal{L} (K^\star)$, $ h^2_M (\frac{1}{\bar{\varepsilon}}, \delta):= h_{sample}(\frac{1}{\bar{\varepsilon}^{''}}, \frac{\delta^{\prime}}{4})$, $\delta^{\prime} = \delta / h_{sample, task} (\frac{2}{\bar{\varepsilon}}, \frac{\delta}{2} ) $, $\bar{\varepsilon}^{\prime} = \frac{\varepsilon}{6 \frac{dk}{r} h_{cost} \bar{J}_{max}}$, $\bar{\varepsilon}^{''} = \frac{\bar{\varepsilon}}{6}$,
Then, when $N \geq \frac{8}{\alpha \bar{\lambda}_i (1 - \eta^2 \bar{h}_H^2)}\log( \frac{2\psi_0}{\epsilon_0}) $, with probability $1- \bar{\delta}$, it holds that, 
\begin{align*}\label{eq:convergence_model_free}
    &\mathcal{L} (K_{N}) - \mathcal{L}(K^\star)  \leq \epsilon_0.
\end{align*}
\end{corollary}

\begin{proof}
    By smoothness property, we have, the meta-gradient update yields,
    \begin{align*}
         & \quad \quad   \mathcal{L} (K_{1}) - \mathcal{L} (K_0)  \leq 
 \langle  \nabla \mathcal{L} (K_0), K_{1} - K_{0} \rangle + \frac{\bar{h}_{{grad}}}{2}\|K_{1} - K_0\|^2_F \\
    &=\langle  \mathcal{L} (K_0), -\alpha \tilde{\nabla}  \mathcal{L}(K_0) \rangle + \frac{\bar{h}_{\mathcal{L},{grad}}\alpha^2}{2}\| \tilde{\nabla}  \mathcal{L}(K_0)\|^2_F\\
    & \leq -\frac{\alpha  }{2}\| \nabla \mathcal{L}(K_0)\|^2_F + \frac{\alpha}{2}\|\ \nabla \mathcal{L} (K_0) - \tilde{\nabla} \mathcal{L}(K_0)\|^2_F + \frac{\bar{h}_{\mathcal{L},{grad}}\alpha^2}{2}\|\tilde{\nabla}  \mathcal{L}(K_0)\|^2_F \\
    & \leq \left( \bar{h}_{\mathcal{L},{grad}}\alpha^2 -\frac{\alpha}{2} \right)\|\nabla \mathcal{L}(K_0)\|^2_F + \left( \bar{h}_{\mathcal{L},{grad}}\alpha^2 + \frac{\alpha}{2} \right)\|\tilde{\nabla} \mathcal{L}(K_0) - \nabla \mathcal{L}(K_0)\|^2_F \\.
    & \leq  - \frac{\alpha}{4}\|\nabla \mathcal{L}(K_0)\|^2_F + \frac{3\alpha}{4} \|\tilde{\nabla} \mathcal{L}(K_0) - \nabla \mathcal{L}(K_0)\|^2_F
    \end{align*}
    The meta-gradient estimation error has been established, it suffices to lower bound the norm $\|\nabla \mathcal{L}(K_0)\|^2_F$ in terms of the initial condition, let $\eta \leq \frac{1}{\bar{h}_H}$,
    \begin{align*}
        \|\nabla \mathcal{L}(K_0)\|^2_F &  = \| \mathbb{E}_{i \sim p }(I - \eta \nabla^2 J^2 (K_0)) \nabla J_i (K_0 - \eta \nabla J_i (K_0 )) \|^2_F , \\
        & \geq \| \mathbb{E}_{i \sim p }  \nabla J_i (K_0 - \eta \nabla J_i (K_0 )) \|^2_F  -  \| \mathbb{E}_{i \sim p } \eta \nabla^2 J_i (K_0) \nabla J_i (K_0 - \eta \nabla J_i (K_0 )) \|^2_F \\
        & \geq   (1 - \eta^2 \bar{h}_H^2 )  \| \nabla J_i (K_0 - \eta \nabla J_i (K_0 )) \|^2_F \\
        & \geq  \mathbb{E}_{i \sim p} \left[\lambda_i(1 - \eta^2 \bar{h}_H^2 ) \left(J_i (K_0 - \eta \nabla J_i (K_0 )) - J_i (K^*_i)\right)\right] \\
        & \geq \mathbb{E}_{i \sim p} \left[\lambda_i(1 - \eta^2 \bar{h}_H^2 ) \left(J_i (K_0 - \eta \nabla J_i (K_0 )) - J_i (K^{\star} - \eta \nabla J_i (K^\star))\right)\right]  \\
        & = \bar{\lambda}_i (1 - \eta^2 \bar{h}_H^2 ) \left[\mathcal{L}(K_0) - \mathcal{L} (K^\star)\right] .
    \end{align*}
    Plugging the above into the expression, we get 
    \begin{align*}
        & \quad \quad   \mathcal{L} (K_{1}) - \mathcal{L} (K^\star) \\
        & \leq \left( 1- \frac{\alpha \bar{\lambda}_i (1 - \eta^2 \bar{h}_H^2 ) }{4} \right)  \left[ \mathcal{L} (K_0 ) - \mathcal{L} (K^\star)\right] + \frac{3\alpha}{4} \|\tilde{\nabla} \mathcal{L}(K_0) - \nabla \mathcal{L}(K_0)\|^2_F, 
    \end{align*}
    let $ \psi_0 : = \mathcal{L} (K_0 ) - \mathcal{L} (K^\star) $, and $\bar{\varepsilon} := \frac{\bar{\lambda}_i (1 - \eta^2 \bar{h}_H^2)\psi_0}{6}$, additionally,
    \begin{align*}
      | \mathcal{T}_n |& \geq  h_{sample, task} (\frac{2}{\bar{\varepsilon}}, \frac{\delta}{2}) , \\
        \ell & \geq  \max\{h^1_{\ell} (\frac{1}{\bar{\varepsilon}^{\prime}}, \delta^{\prime}) ,  h^2_{\ell, grad} (\frac{12}{\bar{\varepsilon}}), h^2_{\ell, var}(\frac{12}{\bar{\varepsilon}}, \delta^{\prime})\} , \\
        r & \leq \min \{ h^2_r(\frac{6}{\bar{\varepsilon}}),  h^1_r (\frac{1}{\bar{\varepsilon}}) \}  , \\
        M & \geq \max\{ h^2_M (\frac{1}{\bar{\varepsilon}}, \delta), h^1_M (\frac{1}{\bar{\varepsilon}^{''}}, \frac{\delta}{4})   \} , 
    \end{align*}
    where $ h^2_M (\frac{1}{\bar{\varepsilon}}, \delta):= h_{sample}(\frac{1}{\bar{\varepsilon}^{''}}, \frac{\delta^{\prime}}{4})$, $\delta^{\prime} = \delta / h_{sample, task} (\frac{2}{\bar{\varepsilon}}, \frac{\delta}{2} ) $, $\bar{\varepsilon}^{\prime} = \frac{\varepsilon}{6 \frac{dk}{r} h_{cost} \bar{J}_{max}}$, $\bar{\varepsilon}^{''} = \frac{\bar{\varepsilon}}{6}$, then, when $N \geq  \frac{8}{\alpha \bar{\lambda}_i (1 - \eta^2 \bar{h}_H^2)}\log( \frac{2\psi_0}{\epsilon_0}) $, we can apply a union bound argument to arrive at $\mathcal{L}(K_N) - \mathcal{L}(K^\star) \leq \epsilon_0 $ with probability at least $1 - N \delta$. Letting $\bar{\delta} = \frac{1}{N} \delta$ completes the proof.
\end{proof}